\documentclass[11pt]{article}
\usepackage{amsbsy}
\usepackage{epsfig}
\usepackage{url}
\usepackage{color}
\usepackage{multirow}
\usepackage{amsmath}
\usepackage{amsfonts}
\usepackage{braket}
\usepackage{amsthm}
\usepackage{amssymb}

\usepackage{multirow}
\usepackage{multicol}
\usepackage{verbatim}
\usepackage{url}
\usepackage{comment}
\usepackage{mathtools}
\usepackage{epsf,pgf,graphicx}
\usepackage{color}
\usepackage{tikz,pgf}
\usepackage{diagbox}
\usepackage{makecell}
\usepackage{braket}
\usepackage{mathtools}
\usepackage{xcolor}
\usepackage{subfig}
\usepackage{hyperref}
\usepackage{ulem}
\usepackage{arydshln}
\usepackage{graphicx}
\usepackage{fancyhdr}

\title{\bfseries
 Measures to characterise Approximate Mutually Unbiased Bases
}

\author{
Ajeet Kumar$^{1}$,
Uditanshu Sadual$^{2}$
}

\date{December 14, 2025}
\makeatletter
\newcommand{\FullFootnoteRule}{%
  \kern -3pt
  \hrule width\textwidth height 0.4pt
  \kern 2.6pt
}
\makeatother

\newtheorem{theorem}{Theorem}[section]
\newtheorem{corollary}{Corollary}[theorem]
\newtheorem{lemma}[theorem]{Lemma}
\newtheorem{remark}{Remark}

\newtheorem{definition}{Definition}

\newtheorem{note}{Note}

\usepackage[backend=bibtex]{biblatex}
\begin{document}

\maketitle

\thispagestyle{plain}
\makeatletter
\renewcommand{\footnoterule}{\FullFootnoteRule}
\makeatother
\footnotetext[1]{$^{}$Indian Statistical Institute, Kolkata}

\footnotetext[2]{$^{}$Indian Institute of Technology-Delhi,New Delhi}



\begin{abstract}
Mutually Unbiased bases has verious application in quantum information procession and coding theory. There can be maximum $d+1$ MUBs in $\mathbb{C}^d$ and $d/2+1$ MUBs in $\mathbb{R}^d$.  But , over $\mathbb{R}^d$ MUBs are known to be non existent when $d$ is odd and for most of the other even $d$ there are mostly 3 Real MUBs. In case of $\mathbb{C}^d$ the construction for complete set of MUBs are known for only Prime Power dimensionn. Thus in general large set of MUBs are not known, particularly for composite dimensions which are not of the form of prime powers.

Because of this, there are many constructions of Approximate version of MUBs. In this paper we make an attempt to define certain measures to characterise the AMUBs. Our construction of measures derives its inspiration from the applications of MUBs, and based on them, we define certain quantifiable measures, which are can be computed and gives estimates of how close the Approximate MUBs are to the MUBs. We use geometric interpretation, projective design features of MUBs and applications  like Optimal State determination and  Entropic Uncertainty of MUBs.

We show generic relationship between these measures and show that it can be evaluated for APMUBs without known the exact construction details, thereby showing that definition of APMUB is sufficient completely characterise it. We also evaluate these measure for an interesting class of AMUBs called Weak MUBs and certain AMUBs constructed using RBDs.

\end{abstract}

\section{Introduction}

Mutually Unbiased Bases (MUBs) are fundamental structures over Hilbert spaces that have found numerous applications in quantum information theory, including quantum tomography, state estimation, quantum key distribution, dense coding, and entanglement swapping. When constructed over $\mathbb{R}^d$, they yield {Real MUBs}, which have deep connections with quadratic forms \cite{cameron1991quadratic},  association schemes \cite{lecompte2010equivalence, delsarte1973algebraic}, equiangular lines, equiangular tight frames, fusion frames over $\mathbb{R}^d$ \cite{bodmann2018maximal}, mutually unbiased real Hadamard matrices, and bi-angular vectors over $\mathbb{R}^d$ \cite{holzmann2010real,best2013biangular,best2015mutually,kharaghani2018unbiased}, as well as applications in the construction of codes \cite{calderbank1997?4}.  

Further, several works have explored their geometric and combinatorial links with objects such as polytopes and projective planes \cite{bengtsson2005mubs,bengtsson2005mutually,saniga2004mutually,saniga2005hjelmslev,appleby2014galois}.

It is known that the maximum number of MUBs in $d$ dimensions is $d+1$ over $\mathbb{C}^d$, and $d^2+1$ over $\mathbb{R}^d$. The \cite{wootters1989optimal}  showed through a construction that for dimension of the form of  prime power $(d = p^m)$ such $(d+1)$ MUBs exist.  In fact till now the complete sets of MUBs over $\mathbb{C}^d$ are known to exist only when $d$ is a power of a prime. These are  precisely the dimensions for which finite affine planes (and hence finite projective planes) exist. This leads naturally to conjectures connecting the existence of complete sets of MUBs with that of finite projective planes. For a composite dimension $d = p_1^{n_1}p_2^{n_2}\cdots p_r^{n_r}$, the lower bound on the number of known MUBs is $\min(p_1^{n_1}, p_2^{n_2}, \ldots, p_r^{n_r}) + 1$. Constructing large numbers of MUBs for such composite dimensions has proven elusive, even over $\mathbb{C}^d$, and the situation is more restrictive over $\mathbb{R}^d$. Indeed, large sets of real MUBs are non-existent for most dimensions \cite{boykin2005real}; only when $d = 4s$ with $s > 1$ do we have $d^2 + 1$ MUBs, whereas for most non-square $d$, only two real MUBs are known.

To address this limitation, several constructions of {Approximate Mutually Unbiased Bases} (AMUBs) have been proposed over both $\mathbb{C}^d$ and $\mathbb{R}^d$, aiming to generate large families of nearly unbiased bases in all dimensions. When such constructions are specifically over $\mathbb{R}^d$, we refer to them as {Approximate Real MUBs (ARMUBs)} \cite{kumar2021approximate,kumar2022resolvable,yang2021constructions}. However, the term ``approximate'' has been interpreted differently across works, often inspired by the definitions first used in Approximate SIC-POVM literature. Various relaxations of the MUB condition have appeared such as $|\langle u, v \rangle| = \frac{1+o(1)}{\sqrt{d}}, |\langle u, v \rangle| = \frac{2+o(1)}{\sqrt{d}}$, $ \mathcal{O}{(\frac{\log d}{\sqrt{d}})}$, $ \mathcal{O}(\frac{1}{d^{1/4}})$, and $ \mathcal{O}(\frac{1}{\sqrt{d}})$. \cite{wang2018two,cao2016more,sripaisan2020approximately,li2015constructions,yang2021constructions}.

In this direction, we define two types of approximate mutually unbiased bases (AMUBs): the {$\beta$-AMUBs} and the {APMUBs}. 

The former, $\beta$-AMUBs, are approximate MUBs such that for any pair of vectors \( u, v \) belonging to different orthonormal bases, we have
$
|\langle u, v \rangle| = \frac{\beta}{\sqrt{d}},
$
where \( \beta \) is bounded by some constant for all \( d \). 
where \( \beta \) is bounded by some constant for all \( d \). 
Later, we define \textit{(APMUBs)} 
as those AMUBs satisfying
$
|\langle u, v \rangle| = \frac{\beta}{\sqrt{d}}, \hspace{0.5mm} where\ \beta = 1 + O(d^{-\lambda}) \leq 2,
$
for some constant \( \lambda > 0 \), and
$
\langle u, v \rangle \in \left\{ 0, \frac{\beta}{\sqrt{d}} \right\}
\quad
$ for all d.

Among different definition of Approximate MUBs, by different authors, the one that 
closely resembles with ours is the one defining Approximate MUBs as 
$
|\langle \psi_i^d | \psi_j^{d'} \rangle| \le \frac{1 + o(1)}{\sqrt{d}} .
$
Here we can interpret \( \beta = 1 + o(1) \). Since the term \( o(1) \to 0 \) as \( d \) increases, 
but the authors are not specifying how the term \( o(1) \) approaches zero. 
For example \( d^{-\lambda} = o(1), \lambda > 0 \), also \( (\log(d))^{-1} = o(1) \) 
or so does \( (\log(\log(d)))^{-1} = o(1) \), the later two approaching zero much slowly 
with increase in \( d \) than the first one. To qualify for APMUB, as per our definition, 
we have included the criteria \( \beta = 1 + \mathcal{O}(d^{-\lambda}) \), for some \( \lambda > 0 \), 
hence the term \( \mathcal{O}(d^{-\lambda}) = o(1) \) would approach to \( 0 \), reasonably fast 
with increase in \( d \). In other word, the convergence criteria for APMUB is more stringent 
than for AMUBs, with \( \beta = 1 + o(1) \). Thus \( \beta \) would approach to \( 1 \), reasonably 
fast for APMUB with increase in \( d \). In fact we would see that for most of our 
\textit{APMUB} construction, \( \lambda = \frac{1}{2} \), thus indicating decent rate of convergence of 
\( \beta \to 1 \) with increase in \( d \). Further, we have impose the condition APMUB that, 
\( \beta < 2 \), so that, APMUB is good approximation to MUBs even in lower dimension, and at no place, 
the inner product between the vectors of different basis, is more than twice that of MUBs.

All these definition of AMUBs, the focus has been on the \( |\langle u|v\rangle| \) and to what extent 
the condition of MUBs \( |\langle u|v\rangle| = \frac{1}{\sqrt{d}}\) can be relaxed. 
But to ascertain the closeness of Orthonormal bases to MUBs we need to define quantifiable measures 
which can give indications about deviations in theoretical/practical results in a setup when AMUBs are used 
from that when MUBs are used. Further these measures should be easily relatable to the quantities and 
properties under study using MUBs. And such measures should also be able to indicate the dependencies 
of the parameters of constructed AMUBs, to the ``closeness" of MUBs. As such information will enable one 
to get appropriate informations about parameters which can be altered to get AMUBs closer to MUBs.

For example Bengtsson~\cite{durt2010mutually} has considered distance measure between two orthonormal bases 
to quantify their closeness with MUBs. If the the two bases are MUBs then they have distance 1 and 
if they are identical then the distance measure is 0. In this distance measure, conceptually MUBs are most distant bases.

In this paper we identify other such measures for AMUBs which when known can give indications about the 
quality of AMUBs based on certain quantifiable aspects about the deviations from MUBs, if AMUBs are used 
instead of MUBs in the particle situations. We examine certain common applications of MUBs, and then identify 
certain measures which can indicate the deviations and hence the quality of AMUBs.

In particular, MUBs correspond to
tight projective 2-designs, form optimal packings in complex projective space, saturate
Welch-type bounds, and yield extremal configurations in association schemes. When one
replaces these exact structures with approximate analogues, the relevant quantities that determine operational performance are not the overlaps themselves but rather various
aggregate expressions built from them, such as quadratic Hilbert-Schmidt correlations,
deviations from tight-frame identities, and perturbations of the Gram matrices that
encode the geometry of basis projectors.

Motivated by these observations, we analyse several canonical contexts in which exact
MUBs are known to satisfy extremal or optimality conditions. These include the
Hilbert-Schmidt embedding of rank-one projectors, the Grassmannian geometry of
subspace packings, spherical 2-design conditions, and optimal state reconstruction formulas used in quantum state tomography, the entropy uncertainty relations and quantum key distribution in cryptographic protocols. In each case, we extract natural scalar quantities that measure the departure of a collection of bases from satisfying the identities or
equalities characteristic of perfect mutual unbiasedness. These quantities serve as
intrinsic, basis-independent diagnostics for the quality of an AMUB construction, and
they allow meaningful comparison between constructions arising from algebraic,
combinatorial, or analytic techniques.

The main contributions could be seen in a way that we introduce a family of geometric and operator-theoretic measures that quantify the
deviation of an AMUB family from the defining identities of exact MUBs. These measures
are derived from structural equalities satisfied by MUB projectors, such as orthogonality of traceless components, tight-frame relations, and Grassmannian distance
characterisations, and therefore provide mathematically principled indicators of
approximate unbiasedness. We apply the proposed measures to several representative applications in which MUBs play an extremal role, including projective 2-design conditions, optimal state
reconstruction formulas, and pairwise separation in projective packing problems. Through these analyses, we identify precisely which aspects of MUB optimality are preserved under approximate unbiasedness and which require stronger control over the approximation
parameters. This clarifies the operational and geometric limitations of AMUBs and provides guidance for constructing AMUB families whose performance is close to that of exact MUBs. Together, these results establish a unified mathematical framework for evaluating the
quality of AMUB constructions through quantities that reflect the structural, geometric, and operational features that make MUBs central to quantum information theory.

\section{Identifying  measures to characterising AMUBs }
\label{sec:motivation}

%
%

Mutually unbiased bases are at the heart of the complementarity principal of Quantum Physics. Over a course of time, the importance of MUBs have grown in Quantum Information processing and quantum cryptography.  In this section we delve into few of the important applications, where MUBs are being used, and we examine the property of MUBs which are relevant for such application. It will provide motivation  in define measures, to quantify the deviation in such property, when in place of MUBs, certain approximate version of them are used.  These quantities  in a way will provide measures  to ascertain the closeness approximate version (AMUBs) with that of MUBs, in a sense that it would be an indicator of deviation in the out come, if in place of MUBs, we use AMUBs.

\begin{itemize}

\item QKD: Quantum Key distribution protocol is one of the first practical application of Quantum System based cryptography. Introduced in 1984 by Bennet and Brassard, uses qubit system for encoding information. It uses MUBs in $d= 2$ dimension. The important aspect in QKD, is the Key  Rate, which referes to the notion of  average number of bits of information recovered from each single qubit sent. Various factor decides the actual Key rate in any practical QKD, but it is well established that if one uses qudits, in place of qubits, key rate is higher. Further the quantum bit error rate (QBER) threshold, which is related to the security threshold of the protocol, increases where all the three MUBs are used in $d=2$, called six state protocol akin to normal BB84 protocol. In addition to this six state protocol has more noise tolerance than normal BB84 protocol. In general using more number of MUBs in higher dimension, one can theoretically show that it has higher bit rate and more tolerant to the errors. Hence Number of  MUBs is important measure in any dimension.

We consider a prepare-and-measure QKD protocol in which the legitimate parties
(Alice and Bob) use \(k\) orthonormal bases \(\{M_1,\dots,M_k\}\) of \(\mathbb{C}^d\).
Alice chooses uniformly at random a basis \(M_\ell\) and a basis vector from that
basis to prepare and send; Bob measures in a basis chosen uniformly at random from
the same \(k\) bases. We analyse (i) the average number of classical bits obtainable
per transmitted system (the raw key rate per signal), and (ii) the disturbance
induced by a simple intercept-resend attack in which an adversary Eve measures the
signal in a basis chosen uniformly from the same \(k\) bases and resends her outcome.

\textbf{Fact1:}[Raw bits per transmitted system]
In the protocol above, each time Alice and Bob use the same basis they obtain
\(\log_2 d\) classical bits from the single \(d\)-level system \cite{bechmann2000quantum2}. Since their bases
coincide with probability \(1/k\), the average number of raw bits obtained per
transmitted system equals
\[
R_{\text{raw}} \;=\; \frac{1}{k}\,\log_2 d.\] \cite{cerf2002security} 

In particular, for fixed \(k\) the raw bits per transmitted system is a strictly
increasing function of the local dimension \(d\).

\begin{proof}
When Alice and Bob have used the same basis, the measurement outcome at Bob is
one of \(d\) equiprobable labels and so carries \(\log_2 d\) bits. The probability
that Bob's basis equals Alice's basis is \(1/k\) (uniform independent choices),
hence the average number of bits per transmitted system is \((1/k)\log_2 d\).
Monotonicity in \(d\) is immediate because \(\log_2 d\) is strictly increasing.
\end{proof}

\textbf{Fact2:}[Detectability increases with \(d\) and \(k\)] \cite{fuchs1996quantum}
Under the intercept-resend attack described above, condition on the ``sifted''
events (Alice and Bob used the same basis). Then the sifted error rate introduced
by Eve equals
\[
E_{\mathrm{sift}} \;=\; \frac{k-1}{k}\left(1-\sum_{j=1}^d p_j^2\right),
\]
where for fixed Alice state \(|\psi\rangle\in M_\ell\) and a (different) basis
\(M_m\) the numbers \(p_j = |\langle \phi_j|\psi\rangle|^2\) are the overlap probabilities
with vectors \(\{|\phi_j\rangle\}_{j=1}^d\) of \(M_m\). In the special case of
mutually unbiased bases (MUBs) we have \(p_j=1/d\) for all \(j\), and hence
\[
E_{\mathrm{sift}}^{\mathrm{MUB}} = \frac{k-1}{k}\Big(1-\frac{1}{d}\Big).
\]
Therefore, for the intercept-resend attack the sifted error \(E_{\mathrm{sift}}^{\mathrm{MUB}}\)
is strictly increasing in both \(d\) (for \(d\ge 2\)) and \(k\) (for \(k\ge 2\)).\cite{cerf2002security}

\begin{proof}
Fix Alice's prepared state \(|\psi\rangle\in M_\ell\). If Eve measures in the same
basis as Alice then no disturbance is introduced. If Eve measures in a different
basis \(M_m\) then she obtains outcome \(j\) with probability \(p_j\) and resends
the post-measurement state \(|\phi_j\rangle\) to Bob. When Bob subsequently
measures in Alice's basis, the conditional probability that he recovers the original
label equals \(\sum_j p_j^2\) (product of the probability Eve produced outcome \(j\)
and the probability Bob obtains Alice's label from \(|\phi_j\rangle\)). Thus the
conditional error given Eve chose the wrong basis is \(1-\sum_j p_j^2\). Since Eve
chooses a basis different from Alice's with probability \((k-1)/k\), averaging
over Eve's basis choice yields the stated formula.

For exact MUBs, \(p_j=1/d\) for all \(j\), so
\(\sum_j p_j^2 = d(1/d)^2 = 1/d\), and therefore
\(E_{\mathrm{sift}}^{\mathrm{MUB}} = \frac{k-1}{k}(1-1/d)\). To see monotonicity:
\[
\frac{\partial}{\partial d}\Big(1-\frac{1}{d}\Big) = \frac{1}{d^2} > 0,\qquad
\frac{\partial}{\partial k}\frac{k-1}{k} = \frac{1}{k^2} > 0,
\]
hence \(E_{\mathrm{sift}}^{\mathrm{MUB}}\) increases with both \(d\) and \(k\).
\end{proof}

\textbf{Corollary (BB84 and six-state protocol)} \cite{nielsen2002quantum} \cite{scarani2009security}
For BB84 one has \(d=2, k=2\) and thus \(R_{\text{raw}}=(1/2)\log_2 2 = 1/2\) bit per
transmitted qubit and \(E_{\mathrm{sift}}= \tfrac12(1-1/2)=1/4\) for intercept-resend.
If instead the six-state protocol (three MUBs in \(d=2\)) is used, then
\(k=3\) and the error becomes \(E_{\mathrm{sift}}=\tfrac23(1-1/2)=1/3>1/4\), showing
that the six-state protocol produces a larger detectable disturbance under this
attack and hence a higher tolerance to noise in this attack model.\\

\item Geometrical aspect of MUBs in $\mathbb{C}^d$: 

A useful way to understand mutually unbiased bases is to place quantum states inside a
geometric framework. Any quantum state is represented by a density matrix, that is, a positive semidefinite Hermitian operator with unit trace. The positivity condition ensures that measurement probabilities are always nonnegative, the trace-one condition ensures
they sum to unity, and Hermiticity guarantees that the eigenvalues are real. The set of all such density matrices forms a
convex body of real dimension \(d^{2}-1\). Pure states correspond to its extreme
points: the rank-one projectors \(P = |\psi\rangle\langle\psi|\).

To obtain a linear structure, it is convenient to shift the maximally mixed state
\(\tfrac{1}{d}I\) to the origin. Every state \(M\) is then represented by the traceless
operator
\[
m = M - \frac{1}{d}I.
\]
This transformation places all quantum states inside a real vector space of dimension
\(d^{2}-1\). The natural choice of inner product in this space is the Hilbert--Schmidt
inner product,
\[
m_{1}\cdot m_{2} = \frac{1}{2}\,\mathrm{Tr}(m_{1}m_{2}),
\]
which equips the space with the structure of a Euclidean vector space.

Under this representation, each pure state projector \(P = |\psi\rangle\langle\psi|\)
corresponds to a vector
\[
p = P - \frac{1}{d}I
\]
of fixed length. An orthonormal basis
\(\mathcal{B}=\{|\psi_{i}\rangle\}_{i=1}^{d}\)
therefore determines \(d\) such vectors \(\{p_{i}\}\), which span a \((d-1)\)-dimensional
subspace of this Euclidean space.

The key geometric insight, emphasized in the work of Bengtsson ~\cite{bengtsson2005mutually}, is that the notion of mutual unbiasedness becomes an
orthogonality relation in this vector space. Indeed, for two bases
\(\mathcal{B}_{1}=\{|\psi^{(1)}_{i}\rangle\}\) and
\(\mathcal{B}_{2}=\{|\psi^{(2)}_{j}\rangle\}\), the condition
\[
|\langle\psi^{(1)}_{i}|\psi^{(2)}_{j}\rangle|^{2}=\frac{1}{d}
\]
is equivalent to the geometric orthogonality relation
\[
p^{(1)}_{i}\cdot p^{(2)}_{j}=0.
\]
Thus, each orthonormal basis corresponds to a \((d-1)\)-dimensional ``plane,'' and two
bases are mutually unbiased precisely when the planes they determine are orthogonal.

Geometric Formulation:
The geometric approach to Mutual Unbiasedness operates within the real vector space
of traceless Hermitian operators. This formulation allows the characterization of two
orthonormal bases, $\mathcal{M}_l = \{|\psi_i^l\rangle\}$ and $\mathcal{M}_m = \{|\psi_j^m\rangle\}$,
by measuring the distance between their corresponding projective lines in this space.

For a pure quantum state defined by the projector $P = |\psi\rangle\langle \psi|$, the associated
traceless operator $m$ (relative to the maximally mixed state $\frac{1}{d}I$) is defined as:
\[
m = P - \frac{1}{d}I.
\]

The inner product in this vector space is defined by the trace:
\[
m_1 \cdot m_2 = \frac{1}{2}\mathrm{Tr}(m_1 m_2).
\]
The condition for two vectors to be unbiased
($|\langle \psi_i^l | \psi_j^m \rangle|^2 = 1/d$)
is equivalent to the orthogonality condition
$m_i^l \cdot m_j^m = 0$.

\textbf{Equivalence to the Bengtsson Distance Term}

We now prove that the deviation measure $\gamma^2_{l,m}$, which defines the Bengtsson
distance \cite{durt2010mutually} $D^2$, is directly proportional to the squared
geometric correlation of the associated traceless operators.

The Bengtsson distance $D^2_{l,m}$ is defined based on minimizing the term:
\[
\gamma^2_{l,m} = \sum_{i,j=1}^{d}
\left( \, |\langle \psi_i^l | \psi_j^m \rangle|^2 - \frac{1}{d} \, \right)^2.
\]

The MUB cost term $\gamma^2_{l,m}$ is proportional to
the sum of the squared geometric correlations between the traceless operators associated
with the two bases $\mathcal{M}_l$ and $\mathcal{M}_m$:
\[
\gamma^2_{l,m} = 4 \sum_{i,j=1}^{d} ( m_i^l \cdot m_j^m )^2.
\]

\begin{proof}
We begin by establishing a fundamental relationship between the dot product of
the traceless operators and the deviation from the ideal MUB density $1/d$.

Consider two traceless operators, $m_i^l$ and $m_j^m$, associated with
vectors $|\psi_i^l\rangle$ and $|\psi_j^m\rangle$. We compute their geometric inner product:
\[
m_i^l \cdot m_j^m = \frac{1}{2} \mathrm{Tr}
\left[
\left( P_i^l - \frac{1}{d}I \right)
\left( P_j^m - \frac{1}{d}I \right)
\right].
\]
Expanding the expression:
\[
\begin{aligned}
m_i^l \cdot m_j^m &=
\frac{1}{2}\mathrm{Tr}
\left(
P_i^l P_j^m
- \frac{1}{d} P_i^l
- \frac{1}{d} P_j^m
+ \frac{1}{d^2} I
\right) \\
&= \frac{1}{2}
\left[
\mathrm{Tr}(P_i^l P_j^m)
- \frac{2}{d} \mathrm{Tr}(P)
+ \frac{1}{d}
\right].
\end{aligned}
\]
Using the properties $\mathrm{Tr}(P)=1$, $\mathrm{Tr}(I)=d$,
and $\mathrm{Tr}(P_i^l P_j^m) = |\langle \psi_i^l | \psi_j^m \rangle|^2$, we obtain:
\[
m_i^l \cdot m_j^m
= \frac{1}{2}
\left[
|\langle \psi_i^l | \psi_j^m \rangle|^2 - \frac{1}{d}
\right].
\]
Rearranging this relationship yields:
\[
|\langle \psi_i^l | \psi_j^m \rangle|^2 - \frac{1}{d}
= 2( m_i^l \cdot m_j^m ).
\]
Substituting this identity into the expression for $\gamma^2_{l,m}$:
\[
\begin{aligned}
\gamma^2_{l,m}
&= \sum_{i,j=1}^{d}
\left[ 2( m_i^l \cdot m_j^m ) \right]^2
= 4 \sum_{i,j=1}^{d} ( m_i^l \cdot m_j^m )^2.
\end{aligned}
\]
This completes the proof.
\end{proof}

This proposition proves that the problem of minimizing the MUB cost $\gamma^2_{l,m}$
is exactly equivalent to minimizing the total squared geometric correlation between the
traceless operators representing the two bases. Since the Bengtsson distance is defined as
\[
D^2_{l,m} = 1 - \frac{1}{d-1}\gamma^2_{l,m},
\]
maximizing the distance $D^2_{l,m}$ provides a rigorous and geometrically meaningful criterion for achieving optimal separation of bases within the Grassmannian manifold of quantum states.

\item Optimal state determination:  Another important practical application which was shown for MUBs, which lead to the vigorous study of MUB and construction of complete set of MUBs for prime power dimension, were application in  state determination by \cite{ivonovic1981geometrical,wootters1989optimal}. The general state of an arbitrary quantum system in $d$ dimensional Hilbert Space, is represented by density matrix $W$, which are positive hermitian matrices such that Tr(W)=1.  Such density matrix is specified by $d^2-1$ real parameters.  And using any $d$-dimensional, orthogonal non degenerate projective measurement  one can build statistics, which will yield probabilities having $d$ outcome which sums to 1, hence will yield $d-1$ real numbers which can be related to the density matrix $W$. Hence one would require $\frac{d^2-1}{d-1} = d+1$  independent $d$-dimensional, orthogonal non degenerate projective measurements to get complete information about density matrix $W$. It is convenient to define $Y = W- \frac{I}{d}$ , which is traceless hermitian matrix. Let $T$ denote set of all traceless Hermitian matrix. Note that the set $T$ form a vector space under matrix addition and is equipped with inner product, given by Trace of the product of two such matrices in the set.

It is known that if one can construct $d+1$  mutually unbiased bases  for corresponding to $d$ dimensional Hilbert Space, $\mathbb{C}^d$, then the projective measurements corresponding to these MUBs provide an optimal means of determining the density matrix of ensemble of identically prepared quantum system, each having $d$ orthogonal states. Here optimal is in the sense that the effects of statistical error are minimised. The minimal statistical error is achieved by measuring the probabilities of $d$ outcome  corresponding to the  operators $\{ P_i^{r}\}_{r=1,..d}$ in respect of the $i^{th}$ MUB, which project the $W$ on to $r^{th}$ Eigen state of the $i^{th}$ MUB.  Crucial in the argument of minimal statistical error  is the maximisation of $ Log( Vol(T_1, . . . . T_{d+1}))$ \cite[Equation 6]{wootters1989optimal}  where $T_i$'s $(d-1)$- dimensional subspace of  traceless Hermitian matrices, and .  The set  $T = \{T_1, . . . . T_{d+1}\}$ represent $(d+1)$ subspaces of $(d^2-1)$ dimensional vector space in $\mathbb{R}^{d^2-1}$,  and $Vol(T_1, . . . . T_{d+1}) $ is the volume of an $(d^2 -1)$-dimensional parallelepiped  constructed by choosing  orthonormal basis for each of the $(d+1)$ subspaces  in set $T$ and the elements of these bases be edges of the parallelepiped. As noted in \cite{wootters1989optimal} the volume of this parallelepiped depends only on the geometrical relationships among the subspaces. It is shown \cite[Equation 9]{wootters1989optimal} that if each of the subspaces in set $\{T_1, . . . . T_{d+1}\}$ corrospondes to an MUBs, hence total $d+1$ MUBs then the $ Vol(T_1, . . . . T_{d+1})$ is maximised, which gives minimal statistical error in the determination of the states. Here the crucial aspect is that any pair of Traceless projectors corresponding to different subspaces would be orthogonal if the subspaces corresponding to to them from MUBs. To see this (as illustrated in \cite[Equation 9]{wootters1989optimal}) let $P_i^r = \left(\ket{\psi_i^r}\bra{\psi_i^r} - \frac{1}{d}I \right)$ then $Tr(P_i^r P_j^s)= |\braket{\psi_i^r|\psi_j^s}|^2 - \frac{1}{d} $ and this quantity is zero when $\ket{\psi_i^r}$ and $\ket{\psi_j^s}$ corresponds to MUBs. And as stated in the ensuing comments of \cite[Equation 9]{wootters1989optimal}, if there do not exist $(d+1)$ such  MUBs, then it will not be possible to make the subspaces
$T_1 ,T_2 \ldots T_{d + 1}$ orthogonal to each other, and the problem of maximising the information with $N + 1$ measurements will be considerably more complicated. 

But one can get approximate value of the decrease in the volume of parallelepiped $ Vol(T_1, . . . . T_{d+1})$ by noting that volume of such parallelepiped is given by $\sqrt{Det(V^\dag V)}$, where $V$ is matrix consisting of the columns of units vectors, forming the edges of the parallelepiped, passing through origin.  For this choose any orthonormal basis for each of the spaces corresponding to $T_1, T_2,\ldots T_{n+1}$ and let the elements of the bases be the edge of the parallelepiped.  But now the vectors spanning $T_i$ is constructed  from AMUBs, in the similar manner as constructed from MUBs.  Considering the fact that $P_i^r = \left(\ket{\psi_i^r}\bra{\psi_i^r} - \frac{1}{d}I \right)$, whose magnitude is $Tr(P_i^r P_i^r )= 1- \frac{1}{d}= \frac{d-1}{d}$. Hence Normalizing them so that they are of unit length, we re-define $P_i^r =\sqrt{\frac{d}{d-1}}\, \left(\ket{\psi_i^r}\bra{\psi_i^r} - \frac{1}{d}I \right)$. Note that  $(P_i^r)^\dag = P_i^r$, hence the matrix $V^\dag = V$. Further note that $\sum_{i=1}^d P_i^r  = 0$, thus only $d-1$ of then are linearly independent. Thus one can choose any $d-1$ unit vectors from them.  With this the diagonal elements of $Det(V^\dag V)$ is $1$ and all the off diagonal term is of magnitude $\frac{(\beta^{lm}_{ij})^2-1}{d-1}$ where $|\braket{\psi_i^l|\psi_j^m}| = \frac{\beta^{lm}_{ij}}{\sqrt{d}}$. Note that when we have MUBs then $\beta^{lm}_{ij} = 1$, for $l\neq m$ and $\beta^{lm}_{ij} = 0$ for $ l=m$ and $i\neq j$. Hence, in this case $V^\dag V$ is block diagonal matrix, consisting of $(n+1)$ identical  blocks each of size $(n-1)\times (n-1)$. On the other hand in case of AMUBs, the $\beta^{lm}_{ij} \leq  \tau$, for $l\neq m$ and $\beta^{lm}_{ij} = 0$ for $ l=m$ and $i\neq j$. Hence, in this case also the block diagonals of the matrix $V^\dag V$, consisting of $(n+1)$ identical  blocks each of size $(n-1)\times (n-1)$, but the off diagonal block elements are not zero.

Using these features, one can estimate the  the term corresponding to $Log(Vol(T_1, . .$\newline$  . . T_{d+1}))= \frac{1}{2}Log\left(Det(V^\dag V) \right)$.  Now the expression for determinant can be ascertained by knowing the exact entries of the matrix, as stated above. But the decrease in the volume, compared to the MUBs situation, can be approximately seen by using the notion that $Tr(P_i^r P_i^s)$ is $cos(\theta)$, where $\theta$ is the angle between the projectors constructed belongings to different orthonormal basis. Hence, 

\[cos(\theta)= |\braket{\psi_i^l|\psi_j^m}|^2 - \frac{1}{d} 
\Rightarrow sin^2(\theta) = 1-  (|\braket{\psi_i^l|\psi_j^m}|^2 - \frac{1}{d})^2\] 

\[ \text{Now\ if\ the\ } V_f= Vol(T_1, . . . . T_{d+1}) \text{\ then\ } V_f^2 \approx V_{MUB}^2 \times \Pi_{i,j}^{l,m} \left(1-  (|\braket{\psi_i^l|\psi_j^m}|^2 - \frac{1}{d})^2\right).\]

\[\text
{Hence\ the\ } Log\left( Vol(T_1, . . . . T_{d+1})\right)\approx  Log (V_{MUB})+ \sum_{i,j}^{l,m} Log\left(1-  \left(|\braket{\psi_i^l|\psi_j^m}|^2 - \frac{1}{d}\right)^2\right)\] 
\[ \Rightarrow  Log (V_f/V_{MUB})= \sum_{i,j}^{l,m} Log\left(1-  |\braket{\psi_i^l|\psi_j^m}|^4 - \frac{1}{d^2}+ \frac{2|\braket{\psi_i^l|\psi_j^m}|^2}{d} \right).\]

One can get an further approximation by using the expansion of \[Log(1-x) = - (x +\frac{x^2}{2}+ \frac{x^3}{3}+\ldots)\] 
Hence  $ -Log (V_f/V_{MUB}) \approx  \sum_{i,j}^{l,m} |\braket{\psi_i^l|\psi_j^m}|^4+  \sum_{i,j}^{l,m} |\braket{\psi_i^l|\psi_j^m}|^6 \ldots $

\item Entropy uncertainty measure. Mutually Unbiased bases are the related to the complementarity aspect of the quantum physics. Which in turn are fundamental to the uncertainty principals in quantum mechanics. The   
uncertainty relations have emerged as the central ingredient in the security analysis of almost all quantum cryptographic protocols, such as quantum key distribution and two-party quantum cryptography. 

Entropic uncertainty relations are an alternative way to state Heisenberg's uncertainty principle which are more useful characterisation, as in this case the uncertainty is lower bounded by a quantity that does not depend on the state to be measured.

Here result by Maassen and Uffink is significant which says that for any two orthonormal basis $B_1$ and $B_2$ in $\mathbb{C}^d$, for any pure state $\ket{\psi}$ in $\mathbb{C}^d$ 
$$
\frac{1}{2} \left( H(B_1 \ket{\psi}) + H(B_2 \ket{\psi}) \right) \geq -  Log(k)
$$
where $k = \max \{ |\braket{v_1|v_2}|: v_1\in B_1, v_2 \in B_2\}  $. Thus when $B_1$ and $B_2$ are MUBs then we have $\frac{1}{2} \left( H(B_1 \ket{\psi}) + H(B_2 \ket{\psi}) \right) \geq \frac{Log(d)}{2}$

When we restrict to a pair of basis, the Massen Uffinik result clearly shows that MUBs are most incompatible if we take entropic uncertainty as measure of incompatibility. But there are results which shows that a set of MUBs, in $\mathbb{C}^d$ are not always most incompatible, when considering more than two observables, according to entropic uncertainty. In fact there exist $\sqrt{d}$ many  MUBs in square dimension, which are not as incompatible as set of approximately $(log(d))^4$ bases, uniformly choose at random (Uniformity in sense of Haar Measure). It has been shown \cite{hayden2004randomizing} that if we choose approximately $(log d)^4$ bases uniformly at random, then
$ \frac{1}{m} \sum_{i=1}^m \left( H(B_1 \ket{\psi} \right) \geq log d-3. $ This means that there exist $m = (log d)^4$ bases for which this sum of entropies is very large, i.e., measurements in such bases are very incompatible on the other had there exist MUBs $\sim \mathcal{O}(\sqrt{d}) $ for which $\frac{1}{m} \sum_{i=1}^m \left( H(B_1 \ket{\psi}\right)\geq \frac{log d}{2} $
Thus in order  to define any measure for set of orthonormal basis to get some idea about its closeness to MUBs,  giving some idea about Entropic uncertainty relation, it is better to use Massen Uffinik result related to any pair of Orthonormal basis.

\item Frame potential measure: The application in quantum information processing, related to Error Correction, Random Access code etc are related to the spherical and the projective design. Mutually unbiased bases has found application related to these aspect.  It is well known that  complete set of MUBs for Projective 2 Design \cite{klappenecker2005mutually}. In general  a set of unit-length vectors forms a complex projective $t-$design, if sampling uniformly from the set gives rise to a random vector whose first $2t$ moments agree with the moments of the uniform distribution on the sphere. This property makes designs a useful tool for the derandomization of constructions that rely on random vectors. Since such projective t-design has various applications.  In \cite[Proposition 1]{zhu2016clifford} gives equivalent statements, when set of vectors constitute $t-$ design. One which can be easily computable is  characterisation through $t-$ Frame Potential which is defined as 
$$
\Phi_t(\{ \psi_j \} ) = \frac{1}{r^2} \sum_{j,k} |\braket{\psi_j|\psi_i}|^{2t}
$$
and Dimension of $Sym_t(\mathbb{C}^d)$, which is defined as  $ D_{[t]}={d+t-1 \choose t} $
As set of vector is $t$-design iff $t-$ Frame Potential is equal to $D_{[t]}$. In fact in general $t-$ Frame Potential $\geq D_{[t]} $. And equality holds only if it is $t-$ design.

\item Sparsity of orthonormal bases is an important feature of our construction. To characterize the sparsity of the MUBs/AMUBs/ARMUBs, we arrange the basis vectors as column of $d \times d$ matrix and use the standard measure of sparsity (denoted by $\epsilon$) as ``number of zero elements in the matrix divided by the total number of elements". The columns of the matrix consist of orthonormal bases vectors. Closer the value of $\epsilon$ to unity, more the number of zeros in the matrix. It will be shown that for our construction the sparsity in general varies as $\epsilon =1 -  \mathcal{O}(d^{-\frac{1}{2}})$. As noted previously we call a set of Orthonormal Basis as $\beta$-AMUB (Approximate MUB) if 
$|\braket{v_1|v_2}| \leq \frac{\beta}{\sqrt{d}}$, for any two vectors $v_1, v_2$ belonging to different bases, where $\beta$ is some constant.

\end{itemize}

\subsection{Defining Measures}

Based on the above observation , in order to measure the closeness of a pair of orthonormal bases $M_l$ and $M_m$ or a set of orthonormal bases $\mathbb{M}$ to  MUBs, we proposes following four measures for any pair (or set) of orthonormal basis to ascertain is closeness with a pair of MUBs ( or set of MUBs) . The focus is also on providing minimum number of such quantities as measures to  ascertain the closeness approximate version (AMUBs) with that of MUBs. We then provide the expression of commonly used measures, in terms of these set of measures. And we also indicate that other measures may be constructed, using these set of four measures.

Let $M = \{M_1,M_2,\ldots M_r \}$ be set of AMUBs. 
\begin{enumerate}

	\item Order of Set $M$: Since to get larger key rate, large number of MUBs are desirable, and in absence of large number of MUBs, we may use AMUBs, hence the number of AMUBs in a given dimension, would be an important measure. This simply count the number of AMUBs

	\item The  $t$-Coherance : We note above, in most of them  that the very critical measure to find any pair of   orthonormal basis closeness to the MUBs is the value of $|\braket{\psi_j|\psi_i}|$. Using this we define  its $t$-Coherance for any pair of Orthonormal basis $(l,m)$ as  $$
\Omega_t^{l,m}  =  \sum_{j,i} |\braket{\psi_j^l|\psi_i^m}|^{2t}
$$

	If the Bases are identical then $\Omega_t^{l,l} = d$. On the other hand if the two bases are MUBs then  $\Omega_t^{l,m} = d^{2-t}  $. Note that  $\Omega_1^{l,m}= d,\,\, \Omega_2^{l,m}= 1$ and for $t\geq 2$, a mutually unbiased pair of bases has the smallest $t$-Coherance. Thus we define a measure based on this and using the notation for AMUBs we express it as
$$
\Omega_t = \frac{2}{r(r-1)} \sum_{l > m} \sum_{j,i} |\braket{\psi_j^l|\psi_i^m}|^{2t} 
$$

\item One can define,
     \begin{equation}
        \mathbf{\tau}^{l,m} = \text{max}\left|\frac{1}{\sqrt{d}} - \left|\braket{\psi_j^l | \psi_i^m}\right|\right|, \forall\hspace{1mm}i,j 
     \end{equation}
     Note that, $\tau^{l,m} = \tau^{m,l}$. Extending for the set of orthonormal bases $\mathbb{M}$, we denote $\mathbf{\tau} =\max_{l \neq m}\{ \mathbf{\tau}^{l, m}\}$.
     \vspace{2mm}
	
	\item The spectrum of the set $\Delta$: Let us define the set for any pair of Orthonormal basis, consisting all the different value of $|\braket{\psi_j^l|\psi_i^m}|$. We define this as $$\Delta_{lm}= \{|\braket{\psi_j^l|\psi_i^m}|: \ket{\psi_j^l} \in M_l \,\text{and}\, \ket{\psi_j^m} \in M_m \}$$. And for a set of Orthonormal basis we define $$\Delta = \bigcup_{l,m} \Delta_{l,m}$$.

\end{enumerate}

Now we show that some of the measures which are commonly used in literature can be expressed using measure defined above.

\begin{enumerate}
	\item Bengtsson \cite{durt2010mutually} has considered distance measure between two orthonormal bases to quantify their closeness with MUBs. If the the two bases are MUBs then they have distance 1 and if they are identical then the distance measure is 0.  Distance between orthonormal bases $M_l , M_m$ is defined as:
\begin{equation}
D^2_{l,m}= 1 - \frac{1}{d-1} \sum_{i,j = 1}^d \left(|\braket{\psi_j^l | \psi_i^m}|^2 - \frac{1}{d} \right)^2
\end{equation}

Note that $D^2_{l,m} =D^2_{m,l}$ and $D^2_{l,l} =D^2_{m,m}=0$. If $M_l$ and $M_m$ are mutually unbiased then $D^2_{l,m} = 1 $. In general $0\leq D^2_{l,m} \leq 1$.	In this distance measure, conceptually MUBs are most distant bases.  Hence using above we have
$$D^2_{l,m}= \frac{1}{d-1}( d - \Omega^{l,m}_2)
$$
	
For a pair of MUB, $\Omega^{l,m}_2=1 \Rightarrow D^2_{l,m}=1$ as noted above. For a set for $r$ bases, Average Square Distance (ASD), denoted as $\bar{D}^{2}$ is defined by the Bengtsson \cite{durt2010mutually} as the average of distance over all pair of orthonormal bases in the set. Hence,
\begin{equation}
\bar{D}^2 =\frac{1}{r^2} \sum_{l,m =1}^r D^2_{l,m} = \frac{2}{r(r-1)} \sum_{l<m =1}^r D^2_{l,m}
\end{equation}
	
On occasions, more appropriate measure of the distance, for the set of orthonormal bases ($\mathbb{M}$), would be to define  $\mathbf{D^2} =\max_{l \neq m}\{ \mathbf{D^2}_{l, m}\}$. Note that  $\mathbf{D^2} \geq \bar{D}^2$, whenever the set has more than two bases.

 \item The variance of the inner products between the vectors of 
    $M_l$ and $M_m$ from $\displaystyle\frac{1}{\sqrt{d}}$ can be considered. That is, one can compute the following quantity,
\begin{equation}
 \mathbf{\sigma}^{l,m} = \frac{1}{d} \sqrt{\sum_{i,j=1}^{d} \left(\frac{1}{\sqrt{d}} - \left|\braket{\psi_j^l | \psi_i^m}\right| \right)^2}    
\end{equation}
     
  to get an estimate of closeness of $M_l$ and $M_m$ to MUBs. Note that, $\sigma^{l,m} = \sigma^{m,l}$. Extending this for the set of orthonormal bases ($\mathbb{M}$), one can define $\mathbf{\sigma} =\max_{l \neq m}\{ \mathbf{\sigma}^{l, m}\}$.
  
   Again using the $t$-coherance we get 
   
  $$  \mathbf{\sigma}_{l,m}^2 =\frac{2}{d}\left( 1- \frac{\Omega_{\frac{1}{2}}^{l,m}}{d^{\frac{3}{2}}} \right) 
  $$
  Note that for MUBs $ \Omega_{\frac{1}{2}}^{l,m}= d^{\frac{3}{2}}$ hence $\mathbf{\sigma}_{l,m}^2 = 0 $ and for any pair of orthonormal basis, since $\mathbf{\sigma}_{l,m}^2 \geq 0 \Rightarrow  \Omega_{\frac{1}{2}}^{l,m} \leq d^{\frac{3}{2}}$.

\end{enumerate}

\section{ Certain general Relationship between different measure  }

Now let us introduce some notations from~\cite{kumar2022resolvable}.
A set of orthonormal bases would be denoted as $\mathbb{M} = \{M_1, M_2, \ldots, M_r\}$ (may not be MUBs) of dimension $d$, $\Delta$ denotes the 
set of different inner product values between the vectors from different orthonormal bases. That is, $\Delta$ contains the distinct 
values of $\left|\braket{\psi_i^l | \psi_j^m}\right|$ $\forall i, j \in \{1, 2, \ldots, d\}$ and $l \neq m \in \left\{1, \ldots, r\right\}$. 
It is to be observed that when $\mathbb{M}$ is an MUB, $\Delta$ is a singleton set with the only element $\frac{1}{\sqrt{d}}$. 
However, for the approximate ones, there will be more than one value in the set. 
The notation $\beta$-ARMUB (Approximate Real MUB) has been used~\cite{kumar2022resolvable} in the context of a set of orthonormal bases $\mathbb{M}$, where the maximum value 
in $\Delta$ is bounded by $\frac{\beta}{\sqrt{d}}$, where $\beta$ is some constant. Note that the definition is significant for situations, where $d$ can be increased and $\beta$ remains bounded by some constant. We will call such set of MUBs as $\beta$-AMUB (Approximate MUB) for both complex and real cases. If we need to emphasise something specific in the context of Approximate Real MUBs, then we will call it $\beta$-ARMUB.

Let $M_l$ and $M_n$ be any two orthonormal bases. Consider $\{\ket{\psi_{i}^{l}}:i=1,2,\ldots,d\}$ and $\{\ket{\psi_{j}^{m}}:j=1,2,\ldots,d\}$ to be the corresponding basis vectors. Expressing $\ket{\psi_{i}^{l}}$ as a linear combination of $\{\ket{\psi_{j}^{m}}:j=1,2,\ldots,d\}$, we get,

$$
\ket{\psi_{i}^{l}} = \alpha_{i1}\ket{\psi_{1}^{m}} + \alpha_{i2}\ket{\psi_{2}^{m}} + \ldots + \alpha_{id}\ket{\psi_{d}^{m}}, i = 1,2,\ldots,d
$$
Since Bases consist of unit vectors, i.e $\braket{\psi_{i}^{l}|\psi_{i}^{l}} = 1 \forall l,i$ hence,

$$
\braket{\psi_{i}^{l}|\psi_{i}^{l}} = |\alpha_{i1}|^{2} + |\alpha_{i2}|^{2} + \ldots + |\alpha_{id}|^{2} = \sum_{j=1}^{d}|\alpha_{ij}|^{2} = 1, i = 1,2,\ldots,d
$$
where $\alpha_{ij} = \braket{\psi_{i}^{l}|\psi_{j}^{m}},\forall\hspace{1mm}i,j \in \{1,2,\ldots,d\}$. Hence, 
\begin{equation} \label{eq-2}
\sum_{j=1}^{d}|\braket{\psi_{i}^{l}|\psi_{j}^{m}}|^{2} = 1 , \forall i = 1,2,\ldots,d
\end{equation}

\begin{equation}\label{eq-3}
    \Rightarrow  \sum_{i,j=1}^{d}|\braket{\psi_{i}^{l}|\psi_{j}^{m}}|^{2} = \sum_{i=1}^{d} \left(\sum_{j=1}^{d}|\braket{\psi_{i}^{l}|\psi_{j}^{m}}|^{2} \right) =\sum_{i=1}^{d}1=d
\end{equation}

Using this, we prove the following result, 

\begin{theorem}
\label{Th1}
 Let $M_l$ and $M_m$ be any two orthonormal bases. Consider $\{\ket{\psi_{i}^{l}}:i=1,2,\ldots,d\}$ and $\{\ket{\psi_{j}^{m}}:j=1,2,\ldots,d\}$ to be the corresponding basis vectors. Then for $\delta > 0$ 
\begin{align} 
d^{-\frac{\delta }{2}} \leq \sum_{j=1}^{d} |\braket{\psi_{i}^{l}|\psi_{j}^{m}}|^{2+\delta} \leq  \,\, &1 \leq\sum_{j=1}^{d} |\braket{\psi_{i}^{l}|\psi_{j}^{m}}|^{2-\delta}\leq d^{\frac{\delta }{2}}  \\
d^{1-\frac{\delta }{2}} \leq \sum_{i,j=1}^{d} |\braket{\psi_{i}^{l}|\psi_{j}^{m}}|^{2+\delta} \leq \,\, &d \leq \sum_{i,j=1}^{d} |\braket{\psi_{i}^{l}|\psi_{j}^{m}}|^{2-\delta}\leq d^{1+\frac{\delta }{2}}  
\end{align} 
\end{theorem}

\begin{proof}
From equation \ref{eq-2}, we have,

$$
\sum_{j=1}^{d} |\braket{\psi_{i}^{l}|\psi_{j}^{m}}|^{2} = 1 
$$


 Since, $ 0 \leq | \braket{\psi_{i}^{l}|\psi_{j}^{m}}|  \leq 1,\,\, \forall \,\, i,j $, it follows that for $\delta >0$,

$$
 \sum_{j=1}^{d} \left( |\braket{\psi_{i}^{l}|\psi_{j}^{m}}| \right)^{2-\delta} > 1 ,\text{ and } \sum_{j=1}^{d} \left( |\braket{\psi_{i}^{l}|\psi_{j}^{m}}| \right)^{2+\delta} < 1
$$

Since this is true for all $j$, summing over $j = 1,2..d$ we get for $\delta >0$ 
$$
 \sum_{i,j=1}^{d}|\braket{\psi_{i}^{l}|\psi_{j}^{m}}|^{2-\delta} > d,   \text{and} \sum_{i,j=1}^{d}|\braket{\psi_{i}^{l}|\psi_{j}^{m}}|^{2+\delta} < d
$$

To estimate further bound on  $ \sum_{i,j=1}^{d}|\braket{\psi_{i}^{l}|\psi_{j}^{m}}|^{2\pm \delta}$, we  determine the extremal (stationary) points of $ \sum_{i,j=1}^{d}|\braket{\psi_{i}^{l}|\psi_{j}^{m}}|^{2\pm \delta}$ with respect to the constraint that $\sum_{j=1}^{d}|\braket{\psi_{i}^{l}|\psi_{j}^{m}}|^{2} - 1 = 0$ or equivalently $\sum_{i,j=1}^{d}|\braket{\psi_{i}^{l}|\psi_{j}^{m}}|^{2} - d = 0$ using the method of Lagrange's Multipliers. The stationary points occur at $|\braket{\psi_{i}^{l}|\psi_{j}^{m}}| = \frac{1}{\sqrt{d}},\forall\hspace{1mm}i,j$. The $ \sum_{i,j=1}^{d}|\braket{\psi_{i}^{l}|\psi_{j}^{m}}|^{2 - \delta}$, has a maximum at  stationary point on the other hand $ \sum_{i,j=1}^{d}|\braket{\psi_{i}^{l}|\psi_{j}^{m}}|^{2+ \delta}$,has a minimum at the stationary points for all $\delta > 0$. Therefore,

$$
\sum_{j=1}^{d}|\braket{\psi_{i}^{l}|\psi_{j}^{m}}|^{2-\delta} <  \sum_{j=1}^{d}\left(\frac{1}{\sqrt{d}}\right)^{2-\delta} = d . d^{-1+\frac{\delta}{2}}= d^{\frac{\delta}{2}} \,\,\, \text{and}$$

$$\sum_{j=1}^{d}|\braket{\psi_{i}^{l}|\psi_{j}^{m}}|^{2+\delta} >   \sum_{j=1}^{d}\left(\frac{1}{\sqrt{d}}\right)^{2+\delta} = d. d^{-1-\frac{\delta}{2}} = d^{\frac{-\delta}{2}} 
$$
Now summing over $j= 1,2..d$, we get 
$$
\sum_{i,j=1}^{d}|\braket{\psi_{i}^{l}|\psi_{j}^{m}}|^{2-\delta} <  d^{1+\frac{\delta}{2}},  \text{and}
\sum_{i,j=1}^{d}|\braket{\psi_{i}^{l}|\psi_{j}^{m}}|^{2+\delta} >    d^{1-\frac{\delta}{2}}
$$

Hence, combining above inequalities we get the desired result.
\end{proof}

\begin{corollary}
Let $M_l$ and $M_m$ be any two orthonormal bases. Consider $\{\ket{\psi_{i}^{l}}:i=1,2,\ldots,d\}$ and $\{\ket{\psi_{j}^{m}}:j=1,2,\ldots,d\}$ to be the corresponding basis vectors. Let $f = \sum_{i,j=1}^{d}|\braket{\psi_{i}^{l}|\psi_{j}^{m}}|^{2+\delta} $, then $\forall \hspace{1mm}\delta \in \mathbb{R}$ 

\begin{equation}
d^{1-\frac{| \delta |}{2}} \leq f = \sum_{i,j=1}^{d}|\braket{\psi_{i}^{l}|\psi_{j}^{m}}|^{2+\delta} \leq d^{1 + \frac{| \delta |}{2}}
\end{equation}
\end{corollary}

Based on above, we show following relationship between different measures.
     
 \begin{itemize}
 
       \item Note that, $\sigma^{l,m} \leq \tau^{l,m}$. This can be easily seen as 
        $$\left(d \mathbf{\sigma}^{l,m} \right)^2 = \sum_{i,j=1}^{d} \left(\frac{1}{\sqrt{d}} - \left|\braket{\psi_j^l | \psi_i^m}\right| \right)^2  \leq  \sum_{i,j=1}^{d}\left( \text{max}\left|\frac{1}{\sqrt{d}} - \left|\braket{\psi_j^l | \psi_i^m}\right|\right|\right)^2 =d^{2}(\mathbf{\tau}^{l,m})^{2} $$
        
        \begin{equation} \Rightarrow  \mathbf{\sigma}^{l,m} \leq \mathbf{\tau}^{l,m}\end{equation}
   
   It can be shown that,
        
        \begin{equation}
            1 - \frac{(d+\sqrt{d})^{2}}{d-1}(\mathbf{\sigma}^{l,m})^{2} \leq  D^{2}_{l,m} \leq 1 - \frac{d}{d-1}(\mathbf{\sigma}^{l,m})^{2}
        \end{equation}

    \begin{proof}
    
    Consider $$\sum_{i,j=1}^{d} \left(\frac{1}{d} - \left|\braket{\psi_j^l | \psi_i^m}\right|^{2} \right)^2  = \sum_{i,j=1}^{d} \left(\frac{1}{\sqrt{d}} - \left|\braket{\psi_j^l | \psi_i^m}\right| \right)^2\left(\frac{1}{\sqrt{d}} + \left|\braket{\psi_j^l | \psi_i^m}\right| \right)^2$$
    Note that, $\braket{\psi_i^l | \psi_i^l} = 1$, $\braket{\psi_i^m | \psi_i^m} = 1$ and $0 \leq \left| \braket{\psi_i^l | \psi_i^m} \right|\leq 1$.
    
    $$\Rightarrow \sum_{i,j=1}^{d}\left( \left|\braket{\psi_j^l|\psi_i^m}\right| - \frac{1}{\sqrt{d}}\right)^{2}\left(\frac{1}{\sqrt{d}}\right)^{2} \leq \mathbf{\gamma}^{2} \leq  \sum_{i,j=1}^{d}\left(\left|\braket{\psi_j^l|\psi_i^m}\right| - \frac{1}{\sqrt{d}}\right)^{2}\left(1 + \frac{1}{\sqrt{d}}\right)^{2}$$
    
    
    $$\Rightarrow d \left(\mathbf{\sigma}^{l,m}\right)^{2} \leq \mathbf{\gamma}^{2} \leq \left(d + \sqrt{d}\right)^{2}\left(\mathbf{\sigma}^{l,m}\right)^{2}$$
    
    Where $\mathbf{\gamma}^{2} = \sum_{i,j=1}^{d} \left(\frac{1}{d} - \left|\braket{\psi_j^l | \psi_i^m}\right|^{2} \right)^2 $
    
 Hence $ D^{2}_{l,m} = 1 - \frac{\mathbf{\gamma}^{2}}{d-1}$. Therefore
    
  
\begin{equation}  
 1 - \frac{(d+\sqrt{d})^{2}}{d-1}(\mathbf{\sigma}^{l,m})^{2} \leq  D^{2}_{l,m} \leq 1 - \frac{d}{d-1}(\mathbf{\sigma}^{l,m})^{2}
\end{equation}
 
    \end{proof}

This also implies that for large $d$, we have   
\begin{equation}
  1 - d \cdot (\sigma^{l,m})^{2} \leq \bar{D}^{2}_{l,m} \leq 1 - (\sigma^{l,m})^{2}
\end{equation}

\item If $(\mathbf{\sigma}^{l,m})= \mathcal{O}\left( \frac{1}{d^{\frac{1}{2} + \delta}}\right) $, where $\delta > 0$, then, $D^{2}_{l,m} \rightarrow 1$ as $d \rightarrow \infty$

    And similarly for set of Orthonormal Basis $\mathbb{M}$, it can be shown  that for large $d$,
\begin{equation}\Rightarrow 1 - d \cdot \mathbf{\sigma}^{2} \leq \bar{D}^{2} \leq 1 - \mathbf{\sigma}^{2}_{\text{min}}
\end{equation} 
where, $\mathbf{\sigma}_{\text{min}} = \min_{l \neq m}\{\mathbf{\sigma}^{l,m}\}$.
    
    \begin{proof}
    We have, $\mathbf{\sigma} = \max_{l \neq m}\{\mathbf{\sigma}^{l,m}\}$ and $\bar{D}^{2} = \frac{1}{r^2} \sum_{l,m =1}^r D^2_{l,m}$. Hence,
    
    $$\frac{1}{r^2}\sum_{l,m=1}^{r}(1 - d(\mathbf{\sigma}^{l,m})^{2}) \leq \bar{D}^{2} \leq \frac{1}{r^2}\sum_{l,m=1}^{r}(1 - (\mathbf{\sigma}^{l,m})^{2})$$
    
    $$\Rightarrow \frac{1}{r^2}\left( r^2 - d\mathbf{\sigma}^{2}r^2 \right) \leq \bar{D}^{2} \leq \frac{1}{r^2}\left( r^2 - \mathbf{\sigma}^{2}_{\text{min}}r^2 \right) $$
    
    $$ \Rightarrow 1 - d \cdot \mathbf{\sigma}^{2} \leq \bar{D}^{2} \leq 1 - \mathbf{\sigma}^{2}_{\text{min}}$$
   
    \end{proof}
\end{itemize}

\section{Bounds for Measures of $\beta$-Approximate MUBs}

In general the different measures as enumerated in the above section depends on the details of the construction of AMUBs. In this section we  derive bounds on the measures $\tau,\sigma$ and $\bar{D}^{2}$  for $\beta$-AMUBs and gives their asymptotic variation, which shows that for large $d$ the $\beta$-AMUBs is very close to MUBs. Lets first define $\beta$-AMUB
\begin{definition}
\label{defbetaAMUB}
The set $\mathbb{M} = \{M_1,M_2,\ldots,M_r\}$ will be called $\beta$-AMUB if $ \braket{\psi_i^l| \psi_j^m} \leq \frac{\beta}{\sqrt{d}} $ such that   $\beta$ is some constant independent of $d$. When the bases are real  (i.e., orthogonal bases), we call them $\beta$-ARMUBs. 
\end{definition}

\begin{itemize}

\item $ \tau  =  \max\{\left| \frac{1}{\sqrt{d}} - |\braket{\psi_{i}^{l}|\psi_{j}^{m}}| \right| \} $. Since $|\braket{\psi_{i}^{l}|\psi_{j}^{m}}| \leq \frac{\beta}{\sqrt{d}}$, and let $|\braket{\psi_{i}^{l}|\psi_{j}^{m}}| \geq \frac{\beta_{min}}{\sqrt{d}}$  then we have  $\tau  =  \max\{|  \frac{\beta}{\sqrt{d}}- \frac{1}{\sqrt{d}} |, | \frac{1}{\sqrt{d}} - \frac{\beta_{min}}{\sqrt{d}} |\} $. This implies that,

\begin{equation}
\label{tau-AMUB}
\Rightarrow \tau = 
\begin{cases}
\frac{\beta-1}{\sqrt{d}}, & \text{if }\beta +\beta_{min} \geq 2\\
\frac{1 -\beta_{min} }{\sqrt{d}}, & \text{if }\beta +\beta_{min} < 2
\end{cases}
\end{equation}

\vspace{2mm}

\item We have,
 $$d^{2}\sigma^{2} = \sum_{i,j=1}^{d}\left(\frac{1}{\sqrt{d}} - |\braket{\psi_{i}^{l}|\psi_{j}^{m}}|\right)^{2} = \sum_{i,j=1}^{d}\left(\frac{1}{d} - \frac{2}{\sqrt{d}} |\braket{\psi_{i}^{l}|\psi_{j}^{m}}| +  |\braket{\psi_{i}^{l}|\psi_{j}^{m}}|^{2}\right)$$
 $$ = d + d - \frac{2}{\sqrt{d}}\sum_{i,j=1}^{d} |\braket{\psi_{i}^{l}|\psi_{j}^{m}}| = 2\cdot \left(d - \frac{1}{\sqrt{d}}\sum_{i,j=1}^{d}|\braket{\psi_{i}^{l}|\psi_{j}^{m}}|\right)
$$.


%

We know that, $d^{\frac{1}{2}} \leq \sum_{i,j=1}^{d}|\braket{\psi_{i}^{l}|\psi_{j}^{m}}| \leq d^{\frac{3}{2}}$, 

\begin{equation}\label{sigma-AMUB}
    \Rightarrow 0 \leq \sigma^2 \leq \frac{2}{d}(1-\frac{1}{d})    
\end{equation}

Note that we have not used at any place that Set of  Orthogonal Bases are $\beta$- AMUB. Hence this relations holds for any pair Orthonormal bases, i.e for any pair of orthonormal basis $\sigma \leq \mathcal{O}(d^{-\frac{1}{2}})$. Thus only for those AMPUBs for which $\sigma< \mathcal{O}(d^{-\frac{1}{2}})$, $\sigma$ can be a good measure for quantifying its closeness with MUBs. 

%

\vspace{2mm}

\item The distance measure $\bar{D}^{2}$ is given by, 

$$
\bar{D}^{2} = 1 - \frac{1}{d-1}\sum_{i,j=1}^{d}\left(\frac{1}{d} - |\braket{\psi_{i}^{l}|\psi_{j}^{m}}|^{2}\right)^{2} = 1 - \frac{1}{d-1}\sum_{i,j=1}^{d}\left(\frac{1}{d^2} - \frac{2}{d}|\braket{\psi_{i}^{l}|\psi_{j}^{m}}|^{2} + |\braket{\psi_{i}^{l}|\psi_{j}^{m}}|^{4}\right) $$
$$ = 1 - \frac{1}{d-1}\left(\frac{d^2}{d^2}- \frac{2}{d}d+\sum_{i,j=1}^{d}|\braket{\psi_{i}^{l}|\psi_{j}^{m}}|^{4}\right) = 1 - \frac{1}{d-1}\left(-1+\sum_{i,j=1}^{d}|\braket{\psi_{i}^{l}|\psi_{j}^{m}}|^{4}\right)
$$

Since, $\sum_{i,j=1}^{d}|\braket{\psi_{i}^{l}|\psi_{j}^{m}}|^{4} \geq 1$, and $|\braket{\psi_{i}^{l}|\psi_{j}^{m}}|\leq \frac{\beta}{\sqrt{d}}$ thus,

\begin{equation}
\label{D-AMUB}
    \Rightarrow 1 - \frac{\beta^{4} -1}{d-1} \leq \bar{D}^2 \leq 1 
\end{equation}

\begin{remark}
Hence, for every $\beta-$AMUB, the  $\tau \rightarrow 0$  and $\bar{D}^{2} \rightarrow 1$ as $d$ increases as $\beta $ remain bounded as per our definition of $\beta-$AMUB. Infact even if $\beta $ is not bounded but varies as $\beta = \mathcal{O}(d^{\frac{1}{2}-\delta})$, for some $\delta >0$ then $\tau \rightarrow 0$ as d increases. Similarly if $\beta = \mathcal{O}(d^{\frac{1}{4}-\delta})$, for some $\delta >0$ then $ \bar{D}^2  \rightarrow 1$ as d increases. Various authores have defined Approximate MUBs like $|\braket{\psi_{i}^{l}|\psi_{j}^{m}}| \leq \frac{1+ o(1)}{\sqrt{d}}$, or  $\frac{2+ o(1)}{\sqrt{d}}$, or $\mathcal{O}\left(\frac{1}{\sqrt{d}}\right)$, or $\mathcal{O}\left(\frac{log(d)}{\sqrt{d}}\right)$, or $\mathcal{O}\left(\frac{1}{\sqrt[4]{d}}\right)$. Thus it is clear that for all these definitions of AMUB, $\tau \rightarrow 0$ as $d$ increases, and similarly $\bar{D}^2  \rightarrow 1$ for all of these definition as $d$ increases except when $|\braket{\psi_{i}^{l}|\psi_{j}^{m}}| \leq \mathcal{O}\left(\frac{1}{\sqrt[4]{d}}\right)$, where $\bar{D}^2 $ may not tend to 1.
\end{remark}
\end{itemize}

%
%

%

\section{Expression for Measures of  Almost Perfect MUBs}
In this section we show that the measure of APMUBs, is independent of the manner of construction, and is completely determined by the $\beta$ and $d$. We first give the definition of the APMUB  
\begin{definition}
\label{defAPMUB}
The set $\mathbb{M} = \{M_1,M_2,\ldots,M_r\}$  of orthonormal basis in $\mathbb{C}^d$ will be called Almost Perfect MUBs (APMUBs) if $\Delta  = \left \{0,\frac{\beta}{\sqrt{d}}  \right \}$, i.e., the set 
contains just two values such that   $\beta< 2$ and $\beta = 1+\mathcal{O}(d^{-\lambda})$, $\lambda > 0$. When the bases are real 
(i.e., orthogonal bases), we call them Almost Perfect Real MUBs (APRMUBs). 
\end{definition}

Among different definition of Approximate MUBs, by different authors, the one that closely resembles with ours is the one defining Approximate MUBs as $|\braket{\psi_{i}^{l}|\psi_{j}^{m}}| \leq \frac{1+ o(1)}{\sqrt{d}}$. Here we can interpret $\beta  = 1+ o(1) $.Since the term $o(1)\rightarrow 0 $ as $d$ increases, but the authors are not specifying how the term $o(1)$ approaches zero. For example $d^{-\lambda} =  o(1), \lambda > 0$, also  $ \left(log(d)\right)^{-1} = o(1)$ or so does $\left(log\left(log(d)\right)\right)^{-1} = o(1)$, the later two approaching zero much slowly with increase in $d$ than the first one. To qualify for APMUB, as per our definition, we have included the criteria $\beta= 1+ \mathcal{O}(d^{-\lambda})$, for some $\lambda > 0$, hence the term $\mathcal{O}(d^{-\lambda}) = o(1)$ would approach to 0 , reasonably fast with increase in $d$. In other word, the convergence criteria for APMUB is more stringent than for AMUBs, with $\beta = 1+ o(1)$. Thus $\beta $ would approach to 1, reasonably fast  for APMUB with increase in $d$. In fact we would see that for most of our $APMUB$ construction, $\lambda = \frac{1}{2}$, thus indicating decent rate of convergence of $\beta \rightarrow 1$ with increase in $d$. 
Further, we have impose the condition APMUB that, $\beta < 2$, so that, APMUB is good approximation to MUBs even in lower dimension, and at no place, the inner product between the vectors of different basis, is more than twice that of MUBs.

Note that since $\Delta = \{0,\frac{\beta}{\sqrt{d}}\}$, let us assume that for $\mathfrak{n}_{1}$ many pairs of $(i,j)$, $|\braket{\psi_{i}^{l}|\psi_{j}^{m}}|= 0$ and $\mathfrak{n}_{2}$ many pairs of $(i,j)$,  $|\braket{\psi_{i}^{l}|\psi_{j}^{m}}| = \frac{\beta}{\sqrt{d}}$, where $\mathfrak{n}_{1} + \mathfrak{n}_{2} = d^{2}$. But we have, $\sum_{i,j=1}^{d}|\braket{\psi_{i}^{l}|\psi_{j}^{m}}|^{2} = d \Rightarrow \mathfrak{n}_{1} \cdot 0 + \mathfrak{n}_{2} \cdot (\frac{\beta}{\sqrt{d}})^{2} = \frac{\mathfrak{n}_{2}\beta^{2}}{d} = d \Rightarrow \frac{\mathfrak{n}_{2}}{d^{2}} = \frac{1}{\beta^2}$ and $\frac{\mathfrak{n}_{1}}{d^{2}} = 1 - \frac{\mathfrak{n}_{2}}{d^2} = 1 - \frac{1}{\beta^2}$.

We now show that our definition of APMUB is sufficient to determine the values of $\tau, \sigma$, and $\bar{D}^2$, in terms of $\beta$, without actually going into the details of the manner of  construction. 
%
\begin{itemize}
\item$ \bf{\tau}$ \\
Since $\Delta = \{0,\frac{\beta}{\sqrt{d}}\}$  this implies $\beta_{min} = 0$ and since $\beta = 1 + \mathcal{O}(d^{-\lambda}) < 2$ this implies for APMUB,

\begin{equation}
\label{tau-APMUB}
\tau = \frac{1}{\sqrt{d}}
\end{equation}
 
Therefore, $\tau = \mathcal{O}(d^{-\frac{1}{2}})$ for APMUB. \\

\item$ \bf{\sigma}$ \\ 
We have previously shown that, $\sigma \leq \tau$. Here, we present a stronger result  for $\sigma $ in case of APMUBs. For two orthogonal basis vectors,
$\{\ket{\psi_{1}^{l}},\ket{\psi_{2}^{l}},\ldots,\ket{\psi_{d}^{l}}\}\text{ and }\{\ket{\psi_{1}^{m}},\ket{\psi_{2}^{m}}, . \newline . . . ,\ket{\psi_{d}^{m}}\}$
We have,
$$d^{2}\sigma^{2} = \sum_{i,j=1}^{d}(\frac{1}{\sqrt{d}} - |\braket{\psi_{i}^{l}|\psi_{j}^{m}}|)^{2}$$
Since, $|\braket{\psi_{i}^l|\psi_{j}^m}|$ takes only two values, which are $0$ and $\frac{\beta}{\sqrt{d}}$ respectively, therefore, let us, $|\braket{\psi_{i}^{l}|\psi_{j}^{m}}| = \frac{\beta}{\sqrt{d}}$ for $\mathfrak{n}_{2}$ pairs of $(i,j)$ and $|\braket{\psi_{i}^{l}|\psi_{j}^{m}}| = 0$ for $\mathfrak{n}_{1}$ pairs of $(i,j)$, where $\mathfrak{n}_{1}+\mathfrak{n}_{2}=d^2$.

$$\Rightarrow d^{2}\sigma^{2} = \mathfrak{n}_{2}(\frac{1}{\sqrt{d}} - \frac{\beta}{\sqrt{d}})^{2} + \frac{\mathfrak{n}_{1}}{d} = 1 + \frac{\mathfrak{n}_{2}\beta^2}{d^2} - \frac{2\mathfrak{n}_{2}\beta}{d^2}$$
Since, $\sum_{i,j}^d |\braket{\psi_{i}^l|\psi_{j}^m}|^2 = d$ 
$\Rightarrow \mathfrak{n}_{2}\cdot (\frac{\beta}{\sqrt{d}})^{2} + \mathfrak{n}_{1}\cdot 0 = d \Rightarrow \frac{\mathfrak{n}_{2}\beta^2}{d^{2}} = 1$. Implying for APMUB
\begin{equation}\label{eq-21}
\sigma^{2} = \frac{2}{d}(1- \frac{1}{\beta})    
\end{equation}

Now we have $\beta = 1 + \mathcal{O}(d^{-\lambda}),\lambda>0 \Rightarrow \frac{1}{\beta} = 1 - \mathcal{O}(d^{-\lambda}), \lambda > 0$, hence we get $ \sigma^{2} = \frac{2}{d}(1-(1-\mathcal{O}(d^{-\lambda}))) = \frac{2}{d}\mathcal{O}(d^{-\lambda}) \Rightarrow \sigma = \mathcal{O}(d^{-\frac{1}{2}-\frac{\lambda}{2}}), \lambda>0 \Rightarrow \sigma \rightarrow \mathcal{O}(d^{-\frac{1}{2}-\frac{\lambda}{2}}), \lambda>0$ or equivalently, $\sigma \rightarrow \mathcal{O}(d^{-\frac{1}{2}-\delta})$, for some $\delta>0$, in case of APMUB.\\

\item $\bar{D}^{2}$\\
Using \ref{eq-3} the explicit calculation can be shown as follows,
$$
\bar{D}^{2} = 1 + \frac{1}{d-1} - \frac{1}{d-1}\sum_{i,j=1}^{d}|\braket{\psi_{i}^{l}|\psi_{j}^{m}}|^{4} =  1 + \frac{1}{d-1} - \frac{m\beta^{4}}{(d-1)d^2} = 1 + \frac{1}{d-1} - \frac{\beta^2}{d-1}
$$
as $m = \frac{d^2}{\beta^2}$ hence we have APMUB

\begin{equation}\label{eq-22}
\bar{D}^2 = 1 - \frac{\beta^2 -1}{d-1} 
\end{equation}

Again since $\beta = 1 + \mathcal{O}(d^{-\lambda}),\lambda>0 \Rightarrow, \bar{D}^2 = 1 - \mathcal{O}(d^{-(1+\lambda)})$ in cae of APMUB.
\end{itemize} 

\section{Expression for Measures of  Weak MUBs}

Weak MUBs have been proposed by M. Shalaby et al \cite{shalaby2012weak} In this section we show that the measure of Weak MUBs, is also independent of the manner of construction, and is completely determined by the definition of it, as in the case of $\beta$-APMUBs. We first give the definition of the Weak MUBs as given in Definition IV.1 \cite{shalaby2012weak}, which we restate using the notation in this paper.  To understand Weak MUBs, which has been constructed for composite dimension. Consider $d= p\times q$ where both $p$ and $q$ are some prime. Hence in $\mathbb{C}^p$ there are set of $p+1$ MUBs  and let us denote the set by $M =\{M_0, M_1, M_2\ldots M_p \}$. Similarly in $\mathbb{C}^q$ there are $q+1$ MUBs and let us denote the set by $N =\{N_0, N_1, N_2\ldots N_q \}$. Now consider the set $M = \{M_i \otimes N_j: M_i\in M \text{and}\, N_j \in N \}$. Note that the $|M| = (p+1)(q+1)$. The set $M$ has been defined as weak MUBs.

\begin{definition}
\label{defAPMUB}
The set $\mathbb{M} = \{M_1,M_2,\ldots,M_{r=(p+1)(q+1)}\}$  of orthonormal basis in $\mathbb{C}^d$  with $d= p\times q$ is called Weak MUBs (WMUBs) if for any pair of orthonormal basis $M_l,M_m\,\, l \neq m$ the $|\braket{\psi^l | \psi^m }|$ will falls in one the following category, where  $\ket{\psi}_i^l \in M_l$ and $ \ket{\psi}_j^m \in M_m$ are the basis vectors.

\begin{enumerate}
\item   $| \braket{\psi^l_i | \psi^m_j } |$  
=  $\begin{cases}
  \frac{1}{\sqrt{q} } & \text{for} \,\, q^2 p \, \text{many pair}  \, ( \ket{\psi^l_i} ,\,\ket{\psi^m_j} )  \\
    0 & \text{for remaining} \,\, q^2 p(p-1)\, \text{many pairs}
\end{cases} $

 \item   $| \braket{\psi^l_i | \psi^m_j } |$  
=  $\begin{cases}
  \frac{1}{\sqrt{p} } & \text{for} \,\, p^2q \, \text{many pair}  \, ( \ket{\psi^l_i} ,\,\ket{\psi^m_j} )  \\
    0 & \text{for remaining} \,\, p^2 q(q-1)\, \text{many pairs}
\end{cases} $ 

\item $|\braket{\psi | \psi }|$ = $\frac{1}{\sqrt{pq }}$ for all pairs of $( \ket{\psi^l_i} ,\,\ket{\psi^m_j} )  $

\end{enumerate}

 And when the bases are real  we call them Weak Real MUBs (WRMUBs). 
\end{definition}

 Note that $\Delta  = \left \{0,\frac{1}{\sqrt{p}}, \frac{1}{\sqrt{q}}, \frac{1}{\sqrt{d}}  \right \}$, and when the pair $M_l, M_m$ are of type 3 above, then the pair is mutually unbiased bases. Let us now calculate the $t$-coherence for for the set $\mathbb{M}$ above.
 
 For this note that there are total $\frac{r(r-1)}{2}\,,\, r=(p+1)(q+1)$ many pairs $\{M_l, M_m \}\,\, l\neq m$. Out of total $\frac{1}{2}(r\times q)$ is of type 1 , $\frac{1}{2} (r\times p)$ is of type 2 and $\frac{1}{2}(r \times pq )$ is of type 3. Hence 
 
 $$
\Omega_t = \frac{2}{r(r-1)} \sum_{l > m} \sum_{j,i} |\braket{\psi_j^l|\psi_i^m}|^{2t} = \frac{2}{r(r-1)} \sum_{l > m} (T_1+ T_2+ T_3)
$$
 Where $ T_1 =\frac{1}{2} rq \frac{q^2p}{q^t}$, $T_2 =\frac{1}{2} rp \frac{p^2q}{p^t}$ and $T_3  = \frac{1}{2}rpq \frac{p^2q^2}{p^tq^t}$ are respectively of pairs of type 1, 2 and 3 of above. Hence
 
  $$
\Omega_t = \frac{1}{(r-1)} \left( q \frac{q^2p}{q^t} +p \frac{p^2q}{p^t} + pq \frac{p^2q^2}{p^tq^t}\right)=  \frac{pq(p^{2-t} + q^{2-t} + p^{2-t}q^{2-t})}{p+q+pq} 
$$

The above can be used to calculate the value of $\bar{D}^2$ and $\sigma$.
%

\section{Measures for AMUB constructed using RBD }
In this section, we caraterize the measure  of AMUBs which can be constructed using RBDs.

Lets first analyse the RBD$(X,A)$ used in construction of AMUB, which have constant block size. The constant block size is essential if we want to use Hadamard Matrix of same order (equal to the block size) for all the blocks of the RBD. Here, we refer to the construction idea of~\cite{kumar2022resolvable}. In any dimension $d$ the $\mu$ and the Block Size $k$ and number of blocks $s$ are the most critical parameter which we decides the closeness of AMUBs to that of MUBs. Hence the measures depends on these parameters as would be evident in following. Our assumption is that, the points of RBD, i.e., $|X|$ can be increased without bound whereas, the parameter $\mu$ remains constant. All our constructions will have this property, which justifies asymptotic analysis of the measures of AMUBs.

If  there exist an RBD$(X, A)$ with $|X|= d= k\times s = (q-e)(q+f)$, with $q, e, f \in \mathbb{N}$ and $e, f \leq c$, where $c$ is some constant. If the said RBD$(X, A)$ consists of $r$ parallel classes, each having blocks of size $(q-e)$, then one can construct $r$ many $\beta$-AMUBs in dimension $d$, where $\beta = \mu \sqrt{\frac{s}{k}}$ and $\epsilon = 1 - \frac{1}{q+f}$. Now using the relation we  $d= q^2 +(f-e) q - ef$  we get $\beta= \mu \sqrt{\frac{q+f}{q-e}} = \mu(\sqrt{1+x^{2}} + x)$, where $x = \frac{f+e}{2\sqrt{d}}$. 

Now when  $\mu=1$, we will get $r$ many $\beta$-APMUBs with $\beta = 1 + \mathcal{O}(d^{-\frac{1}{2}} )$ and  $\Delta =\{0, \frac{1}{q-e} \}$, for which measures can immediately be given as

\begin{itemize}
\item  $\mu = 1$. In this situation we get, APMUBs. The exact value of the measures can quickly be computed using the results derived previously for APMUBs, as shown in equations \ref{tau-APMUB}, \ref{eq-21} and \ref{eq-22}. Here $\beta < 2$ for all $d$. 

\begin{equation}
\label{Par-APMUB-tau}
\tau = \frac{1}{q} = \frac{1}{\sqrt{d}}
\end{equation}

\begin{equation}
\label{Par-APMUB-sigma}
 \sigma^2 = \frac{2}{(q-e)(q+f)} \left( 1- \sqrt{\frac{q-e}{q+f}} \right) = \frac{2}{d}\left(1- \frac{\sqrt{(f+e)^2+4d}- (f+e)}{ 2\sqrt{d}}  \right)
 \end{equation}
 
 \begin{equation}
 \label{Par-APMUB-distance}
     \bar{D}^2 =1- \frac{e+f}{(q-e)\left((q+f)(q-e)-1\right) }=  1- \frac{(f+e) \left( (f+e) +\sqrt{(f+e)^2+4d} \right)} {d(d-1)}
 \end{equation}
Hence we have, 
$$\tau = \mathcal{O}(d^{-\frac{1}{2}}) , \sigma = \mathcal{O}(d^{-\frac{3}{4}})\text{ and }\bar{D}^{2} = 1-\mathcal{O}(d^{-\frac{3}{2}}) $$

Further, if there exist a Real Hadamard matrix of order $(q-e)$,  we can construct $r$ many APRMUBs with same above parameters and measures.

\item $\mu > 1$, we get $\beta $-AMUBs, where we can only get an estimate of the measures $\tau, \sigma $ and $\bar{D}^{2}$, but not an exact value, by using the relation $\beta = \mu(\sqrt{1+x^{2}} + x)$, where $x = \frac{f+e}{2\sqrt{d}}$. We get for $\beta$-AMUB 

$$ \tau = \mathcal{O}(d^{-\frac{1}{2}}) , \sigma = \mathcal{O}(d^{-\frac{1}{2}})\text{ and }\bar{D}^{2} = 1-\mathcal{O}(d^{-1}) $$

The exact relationship will depend on the specific RBD and Hadamard Matrices, used for such construction. And also note that, if $e+f \leq \frac{(2 \mu + 1)}{\mu^2} (q-e)$, then rearranging the terms, we obtain
$\mu^2 \left( \frac{q+f}{q-e} \right) < \mu^2 + 2\mu + 1$ $\Rightarrow$ $\beta = \mu \sqrt{\frac{q+f}{q-e}} \leq \mu + 1$.

\end{itemize}

\begin{note}
Observe that, for MUBs, $\bar{D}^{2} = 1$ and $\tau=0=\sigma$. But here, for $\beta$-AMUBs we have $\tau \rightarrow 0$, $\sigma \rightarrow 0$, and $\bar{D}^{2} \rightarrow 1$ as $d$ increases. In fact for the construction as in the above, $\bar{D}^{2} \rightarrow 1$ much faster than $\sigma \rightarrow 0$ and $\tau \rightarrow 0$. The convergence of $\tau \rightarrow 0$ is the slowest. Thus, we can say that, $\tau$ is the most {stringent} among the three measure, to estimate the closeness of orthogonal basis vectors to MUBs.
\end{note}

Not for  RBD$(X,A)$, with $|X| = d =k\times s= (q-e)(q+f)$, but with block size  $(q+f)$, hence having $(q-e)$ blocks in each parallel class, in such a situation, $\mu > 1$, and hence we cannot get APMUBs. However, they can provide AMUBs as in~\cite{kumar2022resolvable}. The result of~\cite[Theorem 4]{kumar2022resolvable} is a particular case of the above, with $e = 0, f = 1$ and $\mu = 2$, were $q+1$ many parallel classes in an RBD was there, each having constant block size of $(q+1)$. In this case, $\beta = 2 \sqrt{\frac{q}{q+1}} = 2 - \mathcal{O}(\frac{1}{\sqrt{d}})$, i.e., the maximum value of the inner product was slightly less than $\frac{2}{\sqrt{d}}$. 

Such RBDs  would be required in situations, for example, where Real Hadamard Matrix of order $(q+f)$ is available but not of order $(q-e)$. In fact, if $(q-e)$ is not a multiple of $4$, then no Real Hadamard Matrix of order $(q-e)$ is available.

 In this situation, using the relation $d= q^2 +(f-e) q - ef$ we get  $\beta = \mu \sqrt{\frac{q-e}{q+f}} = \mu(\sqrt{1+x^{2}} - x)$, where $x = \frac{f+e}{2\sqrt{d}}$, where $\mu$ is the maximum number of points common between any pair of blocks from different parallel classes. The sparsity is $\epsilon = 1 - \frac{1}{q-e}$. Further, if there exists a Real Hadamard Matrix of order $(q+f)$,  then we can construct $r$ many ARMUBs with the same parameters.

%
%
%
%
Since, $\mu > 1$, this implies that, we can never have APMUBs. Here  $\beta = \mu - \mathcal{O}(q^{-1}) = \mu - \mathcal{O}(d^{-\frac{1}{2}})$. Following the similar analysis of parameters, we get, 
$$\tau = \mathcal{O}(d^{-\frac{1}{2}}), \sigma = \mathcal{O}(d^{-\frac{1}{2}})\text{ and } \bar{D}^{2} =1 - \mathcal{O}(d^{-1})$$

Although we do not obtain APMUBs, but still it is a $\beta$-AMUB hence we have  $\sigma \rightarrow 0$, $\tau \rightarrow 0$ and $\bar{D}^{2} \rightarrow 1$ as $d$ increases,  as given by equations \ref{tau-AMUB}, \ref{sigma-AMUB} and \ref{D-AMUB}.


Now in the general case when the block size of the RBD are not constant, we can still get good quality $\beta$-AMUB, provided the the block sizes are nearly $\sqrt{d}$. Here assuming the existence of  an RBD$(X, A)$ with $|X|= d $, with block sizes in set $K =\{q-e, q-e+1,\ldots q,q+1,\ldots q+f \}$ where   $q, e, f \in \mathbb{N}$ and $0\leq e, f \leq c$, where $c$ is some constant and $q = \mathcal{O}(\sqrt{d})$,  then we can construct $r$ many $\beta$-AMUBs in dimension $d$, with $\beta = \mu \sqrt{\frac{q+f}{q-e}} = \mu - \mathcal{O}(d^{-\frac{1}{2}})$, where $\mu$ is the maximum number of points common between any pair of blocks from different parallel classes. The sparsity is $\epsilon = $. Since these are $\beta$-AMUB, the measures are as given in section for $\beta$-AMUB.

\subsection{Some Approximate MUBs  and its measures }

In this section corresponding to few specific construction of $\beta$-AMUB, which are not APMUB, we calculate and present bounds on the measures.

For any composite dimension $d = k \cdot s, k,s \in \mathbb{N} , k\leq s$ such that $\sqrt{\frac{s}{k}} < 2 $, we have shown above that we can construct at least $N(s) +1$ APMUBs.  .


\begin{lemma}
If $d=s(s+f) $, with $s, f \in \mathbb{N}$ and $f \leq$ $c$, where $c$ is some predefined constant, then one can construct $N(s)+1$ many Approximate MUBs with $\beta = 2 \sqrt{\frac{s}{s+f}}= 2 - \mathcal{O}(d^{-\frac{1}{2}})$  and if there exist Hadamard matrix of order $(s+f)$ i.e., $s+f \equiv  \mod{4}$, then one can construct $N(s)+1$ many Approximate Real MUBs with same parameters. The measures  $\sigma = \mathcal{O}(d^{- \frac{3}{4}})$, $\tau = \mathcal{O}(d^{- \frac{1}{2}})$. 
\end{lemma}

 Note that here $\beta = 2 \sqrt{\frac{s}{s+f}}$, and $f \leq s$, hence minimum $\beta$ which can be achieved in this construction is $\sqrt{2}$, which occurs at $f=s$.  Let us now evaluate various parameters related to this particular design $(\check{X},\check{A})$ with $|\check{X}| = s(s+f)$, where $s$ is some power of prime. Since, this is not an APMUB with $|\braket{v|w}| < \frac{\beta}{\sqrt{d}}$, where $\beta = 2\sqrt{\frac{s}{s+f}}, f\leq s$. Using equation \ref{Par-APMUB-tau}, we get $\tau = \frac{1}{\sqrt{d}}$. And to calculate $\sigma^{2}$ and $\bar{D}^{2}$, we require, $\sum_{i,j=1}^{d}|\braket{\psi_{i}^{l}|\psi_{j}^{m}}|$ and $\sum_{i,j=1}^{d}|\braket{\psi_{i}^{l}|\psi_{j}^{m}}|^{4}$. But, for evaluating this exactly, we need the information on the kind of Hadamard matrices used in converting the parallel classes of RBD into orthonormal bases. But, even without knowing this, and using the fact that each block of one parallel class will have one point in common with $s-f$ blocks of other parallel class and two pint in common with remaining $f$ blocks. This translate in the fact that, any  vector of one Basis will have one point in common with $(s+f)(s-f)$ basis vectors other parallel class and one point in common with remaining $(s+f)f$ basis vectors.   Thus we can get a good estimate, viz. as follows.
 
 \begin{itemize}
 \item For $\sigma^2  = \frac{2}{d}( 1- \frac{1}{d\sqrt{d}}  \sum_{i,j=1}^{d}|\braket{\psi_{i}^{l}|\psi_{j}^{m}}|) $ we have
 $$
 (s+f)(s-f)\frac{1}{s+f} \leq \frac{1}{d} \sum_{i,j=1}^{d}|\braket{\psi_{i}^{l}|\psi_{j}^{m}}| \leq f(s+f)\frac{2}{s+f} + (s+f)(s-f)\frac{1}{s+f}
 $$
 
 $$
 \Rightarrow d(s - f) \leq   \sum_{i,j=1}^{d}|\braket{\psi_{i}^{l}|\psi_{j}^{m}}| \leq d(s+f)
 $$
But a better estimate of the upper bound is given by  Theorem 1, which would imply $ \sum_{i,j=1}^{d}|\braket{\psi_{i}^{l}|\psi_{j}^{m}}| \leq d^{\frac{3}{2}}$, Hence 

$$0 \leq \sigma^2  \leq  \frac{2}{d}( 1- \frac{s - f}{\sqrt{d}}) \Rightarrow \sigma = \mathcal{O}(d^{\frac{-3}{4}})$$
 
 \item For $\bar{D}^2 = 1 - \frac{1}{d-1}\left(-1 + \sum_{i,j=1}^{d}|\braket{\psi_{i}^{l}|\psi_{j}^{m}}|^{4} \right)$ , we have 
 
 $$
 (s+f)(s-f)\frac{1}{(s+f)^{4}} \leq \frac{1}{d} \sum_{i,j=1}^{d}|\braket{\psi_{i}^{l}|\psi_{j}^{m}}|^{4} \leq (s+f)(\frac{2}{s+f})^{4} + (s+f)(s-f)\frac{1}{(s+f)^{4}} 
 $$
 
 $$ \Rightarrow  \frac{d(s-f)}{(s+f)^3}  \leq  \sum_{i,j=1}^{d}|\braket{\psi_{i}^{l}|\psi_{j}^{m}}|^{4}  \leq \frac{d(s+15 f)}{(s+f)^3} $$
 
 But a better lower bound is given by Theorem 1, which would imply:  $1\leq  \sum_{i,j=1}^{d}|\braket{\psi_{i}^{l}|\psi_{j}^{m}}| ^{4} $ Hence 
 
 $$  1 - \frac{1}{d-1} \frac{f(13s - f)}{(s+f)^2}  \leq \bar{D}^2  \leq 1 \Rightarrow \bar{D}^2 = 1 - \mathcal{O}(d^{\frac{-3}{2}}) $$
 
\end{itemize} 
 
 Note that, the sparsity $\epsilon = 1 - \frac{1}{s}$. The asymptotic variation of the parameters in terms of $d$ are $\beta = 2 - \mathcal{O}(d^{-\frac{1}{2}})$. The Asymptotic convergence of $\tau \rightarrow 0, \sigma \rightarrow 0$ and $\bar{D}^2 \rightarrow 1$ is same as that of APMUBs! But here $\beta $ converges to 2.
 
  Note that if we use Real hadamard Matrix of order $(s+f)$, then we will get ARMUBs with $\Delta = \{0, \frac{\beta_1}{\sqrt{d}} , \frac{\beta_2}{\sqrt{d}}\}$ and with $\beta_1 ,\beta_2 \leq 2- \mathcal{O}(d^{-\frac{1}{2}}).$
 


Lets now evaluate the measures for another construction, which is interesting as it is independent of Hadamard Conjecture. For $d= (q-1)(q+1)$ we can also construct an RBD$(X,A)$ such that it has $q+1$ many Parallel classes, each having $q-1$ blocks with constant block size $q+1$. Such that blocks from different parallel classes have either one, or two points in common. Though this construction does not gives better MUBs than the previous one, nevertheless it is an interesting to note that for same $d$, we can get different approximate AMUBs with different parameters, which suit some specific purpose. More over this construction is not dependent on Hadamard Conjecture, as we require here $q \equiv 3 \bmod 4$ which when satisfied, Paley Construction guarantees existence of Real Hadamard Matrix of order $q+1$, hence we demonstrate this case for theoretical purpose and for more clarity.
Further this result is almost of the same quality as in~\cite[Theorem 4]{kumar2022resolvable}.

\begin{lemma}
\label{th-AMUB2}
If $d = q^2-1 = (q-1)(q+1)$, such that $q$ is a prime power and if $q \equiv 3 \bmod 4$, then we can construct $q+1$ many ARMUBs with 
$\Delta = \left\{0, \frac{1}{q+1},\frac{2}{q+1} \right\}$, $\beta =  2\sqrt{\frac{q-1}{q+1}}$, 
$\sigma_o^2 \leq \sigma^2 \leq \sigma_o^2 + \delta$, where $\sigma_o^2 = \frac{2}{d} \left(1- \frac{q}{\sqrt{d}} + \frac{1}{2(q+1)} \right)$, 
$\delta = \frac{4}{d\sqrt{d}} \left(1- \sqrt{\frac{q-1}{q+1} }  \right)$, $\tau  = \frac{1}{\sqrt{d}} $ and $\epsilon = 1 - \frac{1}{q-1} = \frac{q-2}{q-1}$.
The asymptotic variation of the parameters in terms of  $d$ are $\beta = 2- \mathcal{O}(d^{-\frac{1}{2}})$, $\sigma = \mathcal{O}(d^{-\frac{3}{4}}) $, $\tau = \mathcal{O}(d^{-\frac{1}{2}})$.
\end{lemma} 

\begin{proof}
We will show the existence of RBD, consisting of $q+1$ parallel classes, each having $q-1$ blocks of size $q+1$, such that between any two blocks from different parallel classes, there 
will be at most two point in common. Consider an affine resolvable $(q^2, q, 1)$-BIBD. There will be $r = q+1$ parallel classes, consisting of $q$ blocks each having $q$ points. 
Any two blocks from different parallel classes will have exactly one point in common. We will remove one point from this design and then rearrange the points in each  parallel 
class to obtain the required RBD. 

Let us first remove the point number $q^2$ from each parallel class of the design. this will make $|X| = q^2 - 1$. However, now one block of each parallel class would consist 
of $q-1$ points and remaining will remain unaffected. Remove the block containing $q-1$ points from each parallel class (total $q+1$ of them) and arrange them in a rectangular 
array of $q+1$ rows and $q-1$ columns, such that the $i$th row contains the block having $q-1$ points from the $i$th parallel class. Note that, each row of the rectangular array 
will be a disjoint collection of $q-1$ points and the array will contain each point from $X$, as $\lambda = 1$ in this design. All these blocks of the parallel class 
containing the point number $q^2$ will be removed. Let us denote the points in the $(q+1) \times (q-1)$ rectangular array by $q_{i,j}, i= 1, 2, \ldots, (q+1); j= 1, 2, \ldots, (q-1)$. 
Now we will insert the $q-1$ points from the $i$th row of the rectangular array, such that $q_{i,j}$ goes to the block having point $q_{i+1,j}$ in the $i$th parallel class. 
We continue this for all $i \leq q$. For $i = q+1$, insert point $q_{q+1,j}$ in the block having point $q_{1,j}$ in the $(q+1)$th parallel class.

Now the design will consist of $|X|=q^2-1$ points, having $q+1$ parallel classes and each parallel class consists of $q-1$ blocks such that the size of each block is $q+1$. 
The above simple procedure will ensure that no two blocks from different parallel classes will have more than two points in common. In fact every block of a parallel class 
will have 2 points in common with only one block of any other parallel class and exactly one point in common with remaining $q-2$ blocks of that parallel class. 
Therefore $\Delta = \left\{ 0, \frac{1}{q+1} ,\frac{2}{q+1}  \right\}$, which implies $\beta = 2\frac{\sqrt{q^2-1}}{q+1} = 2\sqrt{\frac{q-1}{q+1}}$. 
To calculate $\sigma^2$, we have,
\begin{multline*}
\left( \frac{1}{\sqrt{q^2 -1}}-\frac{1}{q+1} \right)^2(q-2)(q+1 )+ \left(\frac{1}{\sqrt{q^2 -1}}-\frac{2}{q+1}\right)^2 (q+1 ) \leq d \times \sigma^2 \\ \leq \left(\frac{1}{\sqrt{q^2 - 1}}-\frac{1}{q+1} \right)^2(q-2)(q+1 )+ \left(\frac{1}{\sqrt{q^2 -1)}}-0\right)^2 (q+1 ).
\end{multline*}
This simplifies to 
$\sigma_o^2 \leq  \sigma^2  \leq \sigma_o^2 + \delta$ where $\sigma_o^2 = \frac{2}{d} \left(1- \frac{q}{\sqrt{d}} + \frac{1}{2(q+1)} \right)$. Further,
$\delta = \frac{4}{d\sqrt{d}} \left(1- \sqrt{\frac{q-1}{q+1} } \right)$, 
$\tau = \max \left\{  \frac{1}{\sqrt{q^2 -1}}, \left|  \frac{1}{\sqrt{q^2 -1}} - \frac{1}{q+1} \right| , \left|  \frac{1}{\sqrt{q^2 -1}} - \frac{2}{q+1} \right| \right\}
= \frac{1}{\sqrt{q^2 -1}} \approx \frac{1}{q}$ since $q \geq 2$. We also have sparsity $\epsilon = 1- \frac{k}{d} =1- \frac{q+1}{q^2 -1} = \frac{q-2}{q-1}$. 

\end{proof}

Let us consider another interesting construction of AMUBsof previous result from \cite{kumar2022resolvable}, on Approximate real MUBs, which again does not satisfy all the conditions to be a Almost Perfect MUBs, but nevertheless provide very interesting family of Approximate Mutually Unbiased Bases. Here we have, $\Delta  = \{ 0, \frac{\beta}{\sqrt{d}} \}$, i.e., biangularity but $\beta \neq 1 + \mathcal{O}(d^{-\lambda})$.


\begin{lemma}
\label{thprop1}
Suppose, there exists a resolvable $(d, k, 1)$-BIBD. Let $d = k(k-1) t+ k$, where
$1< t \in \mathbb{N}$, then one can construct $\left(k t + 1\right)$ many Approximate 
MUBs in $\mathbb{C}^d$ with $\Delta  = \left\{0, \frac{1}{k} \right\}$, 
$\beta = \sqrt{\frac{(k-1)t+1}{k}}$, $\sigma^2  = \frac{2}{d}\left[1-\frac{k}{\sqrt{d}}\right]$, $D^2 = 1- \frac{1}{k^2}$ and  $\tau  < \frac{1}{k}$ and the sparsity $\epsilon = \left(1 - \frac{k}{d}\right)$. 
Further, if a real Hadamard matrix of order $k$ exists, then one can construct ARMUBs in 
$\mathbb{R}^d$ with the same parameters.
\end{lemma}

\begin{proof}
The proof follows from the construction ~\cite{kumar2022resolvable}. All the parameters has been shown here, except $\bar{D}^{2}$, which we calculate here. In order to calculate this, we need $\sum_{i,j=1}^{d}|\braket{\psi_{i}^{l}|\psi_{j}^{m}}|^4$.

$$
\bar{D}^{2} = 1 + \frac{1}{d-1} - \frac{\sum_{i,j=1}^{d}|\braket{\psi_{i}^{l}|\psi_{j}^{m}}|^4}{d-1}
$$
Hence, in order to calculate this, we require $\sum_{i,j=1}^{d}|\braket{\psi_{i}^{l}|\psi_{j}^{m}}|^4 = \frac{dk^2}{k^4}$,

$$
\Rightarrow \bar{D}^{2} = 1 + \frac{1}{d-1} - \frac{d}{k^2(d-1)} = 1 - \frac{d-k^2}{k^2(d-1)} = \frac{d}{d-1}(1-\frac{1}{k^2})
$$
and, as $d$ increases, we have $\bar{D}^{2} = 1 - \frac{1}{k^2}$.

Note that, $\beta = \sqrt{\frac{(k-1)t+1}{k}} = \mathcal{O}(\sqrt{t})$. Hence, it increases with $t$, without any bound, for fixed $k$. Therefore, this is not an APMUB, even though $\Delta = \{0,\frac{1}{k}\}$, which has just two values. For large $k$, we have $\bar{D}^{2} \rightarrow 1$, even for this construction.

\end{proof}

Note that in above case if $t=1$, then one can construct $\left(k  + 1\right)$ many MUBs in $\mathbb{C}^d$ with  
$\Delta  = \left\{ \frac{1}{k} \right\}$, $\beta = 1$, 
$\sigma = \tau = 0 $ and the sparsity $\epsilon = \left(1 - \frac{1}{\sqrt{d}}\right)$

\section{Conclusion}

In this work we introduced a unified framework for quantifying how close a given
collection of orthonormal bases is to forming a set of mutually unbiased bases. 
Rather than focusing solely on the deviation of individual overlaps from the ideal 
value $1/\sqrt{d}$, we identified several structural and operational measures, 
including Hilbert--Schmidt 
correlations of traceless projectors, Grassmannian distance terms, $t$-coherence functions, and variance-type quantities, that collectively 
capture the geometric, algebraic, and information-theoretic features that make MUBs 
extremal in quantum theory. These measures arise naturally from applications in which 
MUBs serve an optimal role, such as optimal state tomography, entropic uncertainty relations, and quantum key distribution.

A central insight of the paper is that many of these measures depend only on the pattern of cross-basis overlaps, and not on the specific construction of the bases 
themselves. This observation enables a clean characterisation of two classes of 
approximate MUBs. For the Almost Perfect MUBs (APMUBs), which we defined by the 
condition $\Delta=\{0,\beta/\sqrt{d}\}$ together with the convergence 
$\beta = 1 + O(d^{-\lambda})$, $\lambda>0$, we showed that all the relevant measures, including the parameters $\tau$, $\sigma$, $D_2$, and the higher-order coherence 
values, and can be computed exactly in terms of $\beta$. In particular, the asymptotic 
behaviour $\sigma$ and 
$D_2$ demonstrates that APMUBs converge rapidly to the 
behaviour of exact MUBs even without requiring the explicit construction details.

We also analysed the measures for the Weak MUBs introduced by Shalaby and 
Vourdas \cite{shalaby2012weak} . Although these bases are constructed for composite 
dimensions by tensoring complete sets of MUBs from prime dimensions, and therefore 
possess a richer and more irregular overlap structure, our computation shows that their 
$t$-coherence and derived geometric quantities depend only on the combinatorial pattern 
of overlaps encoded in the weak MUB definition. This demonstrates that the measures 
we introduce remain robust even when the unbiasedness condition is highly nonuniform, 
capturing not only perfect but also partially structured forms of complementarity.

Finally, by examining physically motivated tasks, such as QKD detectability under 
intercept--resend attacks, the behaviour of entropic uncertainty lower bounds, and 
the sensitivity of state reconstruction volumes, we illustrated how the deviation 
measures directly translate into operational performance gaps between exact MUBs, 
APMUBs, and other approximate families. This operational viewpoint highlights that 
the proposed measures are not merely geometric diagnostics but constitute actionable 
criteria for judging the practical usefulness of approximate constructions.

Overall, the framework developed here unifies the geometric, algebraic, and 
information-theoretic aspects of approximate mutual unbiasedness. It provides a 
systematic method for evaluating both existing constructions such as Weak MUBs and 
new approximate families such as APMUBs, and clarifies precisely which features of 
MUB optimality degrade under approximation and at what quantitative rate. We expect 
that these measures will be valuable tools in future investigations into approximate 
designs, high-dimensional quantum measurements, and resource-efficient quantum 
cryptographic protocols.

\printbibliography

\end{document}